\documentclass[preprint,3p,times,twocolumn]{elsarticle}
\usepackage{amssymb}
\usepackage{amsthm}
\usepackage{feynmp}
\usepackage{amsmath}
\usepackage{caption}
\usepackage{subcaption}
\usepackage{listings}

\lstnewenvironment{code}
    {\lstset{%
        basicstyle=\small\ttfamily,
        escapechar={\#}, 
        columns=fullflexible
    }}
    {}

\newtheorem{mydef}{Definition}
\newtheorem{mytheorem}{Theorem}
\newtheorem{mylemma}{Lemma}
\newtheorem{mycorollary}{Corollary}

\newcounter{bla}

\journal{Computer Physics Communications}

\begin{document}

\begin{frontmatter}

\title{Feynman graph generation and calculations in the Hopf algebra of Feynman graphs}

\author[a]{Michael Borinsky\corref{author}}

\cortext[author] {Corresponding author.\\\textit{E-mail address:} borinsky@physik.hu-berlin.de}
\address[a]{Institutes of Physics and Mathematics\\
Humboldt-Universit\"at zu Berlin\\
Unter den Linden 6\\
10099 Berlin, Germany}

\begin{abstract}
Two programs for the computation of perturbative expansions of quantum field theory amplitudes are provided. 
\textbf{feyngen} can be used to generate Feynman graphs for Yang-Mills, QED and $\varphi^k$ theories. 
Using dedicated graph theoretic tools \textbf{feyngen} can generate graphs of comparatively high loop orders.
\textbf{feyncop} implements the Hopf algebra of those Feynman graphs which incorporates the renormalization 
procedure necessary to calculate finite results in perturbation theory 
of the underlying quantum field theory. \textbf{feyngen} is validated by comparison to explicit calculations 
of zero dimensional quantum field theories and \textbf{feyncop} is validated using a combinatorial identity on 
the Hopf algebra of graphs.
\end{abstract}


\end{frontmatter}

{\bf PROGRAM SUMMARY}

\begin{small}
\noindent
{\em Manuscript Title: Feynman graph generation and calculations in the Hopf algebra of Feynman graphs}                                       \\
{\em Author: Michael Borinsky}                                                \\
{\em Program Title: \textbf{feyngen}, \textbf{feyncop}
 \footnote{Both programs can be obtained from \texttt{https://github.com/michibo/feyncop}} 
}                                          \\
{\em Journal Reference:}                                      \\
{\em Catalogue identifier:}                                   \\
{\em Licensing provisions:none}                                   \\
{\em Programming language: Python}                                   \\
{\em Computer: PC}                                               \\
{\em Operating system: Unix, GNU/Linux}                                       \\
{\em RAM: 64m} bytes                                              \\
{\em Number of processors used: 1}                              
{\em Classification: 4.4 Feynman diagrams}                                         \\
{\em External routines/libraries: nauty [1]}                            \\
{\em Subprograms used: geng, multig (part of the nauty package)}                                       \\
\\
{\em Nature of problem:
Performing explicit calculations in quantum field theory Feynman graphs are indispensable. 

Infinities arising in the perturbative calculations make renormalization necessary. 
On a combinatorial level renormalization can be encoded using a Hopf algebra [2] whose 
coproduct incorporates the BPHZ procedure. 

Upcoming techniques initiated an interest in relatively large loop order Feynman diagrams 
which are not accessible by traditional tools.
}\\
   \\
{\em Solution method:
Both programs use the established \textbf{nauty} package to ensure high performance graph generation at high loop orders. 

\textbf{feyngen} is capable of generating $\varphi^k$-theory, QED and Yang-Mills Feynman graphs and of filtering these graphs for the properties of connectedness, one-particle-irreducibleness, $2$-vertex-connectivity and tadpole-freeness. It can handle graphs with fixed external legs as well as those without fixed external legs. 

\textbf{feyncop} uses basic graph theoretical algorithms to compute the coproduct of graphs encoding their Hopf algebra structure.
}\\
   \\
{\em Running time:
All $130516$ 1PI, $\varphi^4$, $8$-loop diagrams with four external legs can 
be generated, together with their symmetry factor, by \textbf{feyngen} within eight hours and all $342430$ 1PI, QED, vertex residue type, $6$-loop diagrams can be generated in three days both on a standard end-user PC. 
\textbf{feyncop} can calculate the coproduct of all $2346$ 1PI, $\varphi^4$, $8$-loop diagrams with four external legs within ten minutes.
}\\
   \\

\end{small}
\section{Introduction}
The purpose of this paper is to provide two tools for the computation of perturbative expansions of quantum field theory amplitudes.
These expansions typically come in terms of Feynman graphs, each of them corresponding to a term in the expansion.
These graphs carry Hopf algebra structures on them using the partial order provided by subgraphs. In particular the Hopf algebra structure provided by superficially divergent graphs is needed for the process of renormalization which is indispensable for the perturbative evaluation of a finite amplitude in accordance with renormalized Feynman rules.

All singularities appear in the integrand which is assigned to a Feynman graph by the Feynman rules. These integrands  can be studied through the two Symanzik polynomials and the corolla polynomial for the case of gauge theories \cite{BrownKreimer,KreimerSarsvanSuijlekom}. Short distance singularities which need to be eliminated by renormalization correspond to zeros in the first Symanzik polynomial. 
This basic fact underlies the utility of Hopf algebras in the study of quantum field theory.

The two tools provided here hence are a graph generator for Feynman graphs, \textbf{feyngen} and a routine which automates the Hopf algebra structure of those graphs, \textbf{feyncop}. 

Additionally to the study of perturbative quantum field theories and their renormalization in general, these programs are
aimed to be used as input for systematic parametric integration techniques to 
evaluate Feynman amplitudes \cite{BrownKreimer,browntwopoint}.
These upcoming techniques 
initiated an interest in relatively large loop order 
Feynman diagrams. 
Traditional programs like \textbf{QGRAF} \cite{Nogueira1993279} are designed to generate low loop order diagrams and are thereby insufficient for these applications.
Therefore, both programs use the established 
\textbf{nauty} package, described in \cite{McKay81practicalgraph}, to ensure high performance graph generation at high loop orders. 

\textbf{feyngen} is capable of generating $\varphi^k$-theory, 
QED and Yang-Mills Feynman graphs and of filtering these graphs for the 
properties of connectedness, one-particle-irreducibleness, 
$2$-vertex-connectivity and tadpole-freeness. 
It can handle graphs with fixed external legs as 
well as those without fixed external legs. 

This paper is organized as follows:
An introduction to the properties of the Hopf algebra of Feynman graphs is given in sections \ref{chap:basic} and \ref{chap:hopf_algebra}. 
These properties were used to derive theorem \ref{prop:sum_formula}, an identity on the Hopf algebra, suitable to validate the coproduct computation of \textbf{feyncop}. 
The combinatorial proof of this theorem constitutes a simplification of the proof given in \cite{suijlekom2007ren}.
In section \ref{chap:diagram_gen} details to the implementation and validation of \textbf{feyngen} are laid out and in section \ref{chap:cop_comp} the implementation of \textbf{feyncop} is described.

The manuals of the Feynman graph generation program \textbf{feyngen} and the coproduct computation program \textbf{feyncop} are given in
sections \ref{sec:manual_feyngen} and \ref{sec:manual_feyncop}. Furthermore, conclusions are drawn and some further prospects are outlined.

\section{Feynman graphs}
\label{chap:basic}
Viewed from a graph theoretical point, Feynman graphs are edge-colored multigraphs. Some properties and notions for multigraphs are reviewed in \ref{sec:graph_basic}.
For the treatment of the Hopf algebra of Feynman graphs, a Taylor-made definition 
as in \cite{KreimerSarsvanSuijlekom} is convenient.
\subsection{Definition}
Let $G = \left( V, E \right)$ be a multigraph with the 
vertex set $V$ and edge multiset $E$ being disjoint unions of the 
sets and multisets $\Gamma^{[0]}_\text{int}$, $\Gamma^{[0]}_\text{ext}$, 
$\Gamma^{[1]}_\text{int}$ and $\Gamma^{[1]}_\text{ext}$, such that
$$V = \Gamma^{[0]}_\text{int} \cup \Gamma^{[0]}_\text{ext}$$
and
$$E = \Gamma^{[1]}_\text{int} \cup \Gamma^{[1]}_\text{ext}.$$

The vertices in $\Gamma^{[0]}_\text{ext}$ shall have valency $1$ and 
those in $\Gamma^{[0]}_\text{int}$ valency $\ge3$.
The edges in $\Gamma^{[1]}_\text{int}$ must be incident only to vertices $v \in \Gamma^{[0]}_\text{int}$ and 
those in $\Gamma^{[1]}_\text{ext}$ are incident to at least one vertex $v \in \Gamma^{[0]}_\text{ext}$.
The vertices in  $\Gamma^{[0]}_\text{int}$ are called internal vertices and those in $\Gamma^{[0]}_\text{ext}$ 
are called external or source vertices.
The edges in $\Gamma^{[1]}_\text{int}$ are called internal edges and the ones in 
$\Gamma^{[1]}_\text{ext}$ are called external edges or legs.
\begin{mydef}
\label{def_feynman_graph}
A Feynman graph $\Gamma = \left( G, \text{\normalfont{res}} \right)$ is a pair
of a multigraph $G$ with the above properties and a coloring $\text{\normalfont{res}}$, a map
\begin{align}
   \text{\normalfont{res}} : \Gamma^{[1]}_\text{int} \cup \Gamma^{[1]}_\text{ext} \rightarrow \mathcal{R}_E,
\end{align}
which assigns a color or type from a set of allowed edge types $\mathcal{R}_E$ to every edge in $G$. 

The map $\text{\normalfont{res}}$ can be extended to the internal vertices of $\Gamma$,
\begin{align}
\text{\normalfont{res}} : \Gamma^{[0]}_\text{int}\cup\Gamma^{[1]}_\text{int} \cup \Gamma^{[1]}_\text{ext}
\rightarrow \mathcal{R}_V \cup \mathcal{R}_E,
\end{align}
by assigning a vertex type $r_V \in \mathcal{R}_V$ to every internal vertex $v \in \Gamma^{[0]}_\text{int}$, such that
the vertex types are determined uniquely by the edges incident to $v$. The elements 
$r \in \mathcal{R}_V \cup \mathcal{R}_E$ are called the 
allowed residue types of the theory under inspection. Note that internal vertices can be promoted to corollas of equal valence 
by dividing adjacent internal or external
edges into suitable half-edges and assigning a corolla of valence one to an external vertex.

Generally, the Feynman rules restrict the sets $\mathcal{R}_V$ and $\mathcal{R}_E$, such that
only Feynman graphs with certain vertex and edge types are allowed.
For instance, $\varphi^4$ theory has 
one vertex type,
$\mathcal{R}_V^{\varphi^4} = \left\{
\parbox[c][10pt][t]{12.5pt}{\centering
        \begin{fmffile}{phi4_vtx}
        \begin{fmfgraph}(10,10)
\end{fmfgraph}
        \end{fmffile}}
\right\}$,
and one edge type,
$\mathcal{R}_E^{\varphi^4} = \left\{
\parbox[c][10pt][t]{12.5pt}{\centering
        \begin{fmffile}{phi4_edge}
        \begin{fmfgraph}(10,10)
\end{fmfgraph}
        \end{fmffile}}\right\}$,
whereas quantum electro dynamics
(QED) allows one vertex type, 
$\mathcal{R}_V^{\text{QED}} = \left\{
\parbox[c][15pt][t]{17.0pt}{\centering
        \begin{fmffile}{qed3_vtx}
        \begin{fmfgraph}(15,15)
\end{fmfgraph}
        \end{fmffile}}\right\}$,
and two edge types, 
$\mathcal{R}_E^{\text{QED}} = \left\{
\parbox[c][10pt][t]{12.5pt}{\centering
        \begin{fmffile}{qed_edge1}
        \begin{fmfgraph}(10,10)
\end{fmfgraph}
        \end{fmffile}},
\parbox[c][10pt][t]{12.5pt}{\centering
        \begin{fmffile}{qed_edge2}
        \begin{fmfgraph}(10,10)
\end{fmfgraph}
        \end{fmffile}}
\right\}$.
\end{mydef}
This definition differs slightly from the one in \cite{Manchon} and \cite{KreimerSarsvanSuijlekom}, 
because of the additional external vertices in $\Gamma^{[0]}_\text{ext}$. 
These external or source vertices are 
added for simplicity of the representation of graphs as edge lists, 
for the determination of the isomorphism class of a graph and for the 
transition from multigraphs to simple graphs needed as input for the \textbf{nauty} package.

In the following, Feynman graphs will be referred to as graphs or diagrams 
if the distinction from other types of graphs is clear from the context.

\paragraph{Diagrams with fixed external legs}
Usually when handling Green functions in quantum field theory, the 
external edges and the source vertices of graphs are considered  fixed. 
That means, there is an additional bijective map,
\begin{align*}
\delta : \Gamma^{[0]}_\text{ext} \rightarrow \left\{1, \ldots, |\Gamma^{[0]}_\text{ext}|\right\},
\end{align*}
associated to a graph $\Gamma$, giving a unique numbering of the external vertices. 
This map also induces a unique numbering on the external edges, because 
every external vertex is incident to at least one external edge.
A graph with fixed external edges and vertices is referred to as leg-fixed graph and 
one without fixed external edges and vertices as non-leg-fixed graph.
\paragraph{Isomorphism classes}
The concepts of adjacency, incidence and connectedness apply 
transparently to Feynman graphs using the appropriate properties 
of the underlying multigraph, care must be taken though considering isomorphic classes 
of Feynman graphs.
Two graphs $\Gamma=(G,\text{res})$ and $\Gamma'=(G',\text{res}')$ are isomorphic if there is an isomorphism $\phi$ 
between the two underlying multigraphs $\phi: G \rightarrow G'$, such that $\phi$ preserves the 
edge and vertex types: $\text{res} \circ \phi = \text{res}'$.

For leg-fixed graphs, the isomorphism $\phi$ also needs to preserve the 
numbering of the legs: $\delta \circ \phi = \delta'$.

Figure \ref{fig:fg_props_example} depicts some examples of Feynman 
graphs with the corresponding orders of the automorphism groups for 
the leg-fixed (lf) and the non-leg-fixed (nlf) case.
\paragraph{Subgraphs of Feynman graphs}
A Feynman graph $\gamma$ is a subgraph of a graph $\Gamma$, $\gamma \subseteq \Gamma$, if 
$\gamma^{[1]}_\text{int} \subseteq \Gamma^{[1]}_\text{int}$, $\gamma^{[0]}_\text{int} \subseteq \Gamma^{[0]}_\text{int}$
and every vertex $v \in \gamma^{[0]}_\text{int} \subseteq \Gamma^{[0]}_\text{int}$ has the same 
vertex type in $\gamma$ as in $\Gamma$: $\text{res}_\gamma(v) = \text{res}_\Gamma(v)$. 
That means, possible deficiencies of incident edges in a
subgraph are fixed by adding additional external legs.
Note that a subgraph is given uniquely by its internal edges and vertices. 
The external edges and vertices can be reconstructed from missing 
edges incident to a vertex to fulfill the vertex type requirement.

A canonical ordering of the external legs of subgraphs
cannot be defined easily. Therefore, subgraphs are 
considered as non-leg-fixed graphs in the scope of this thesis.
Bose symmetry ensures though that in the sum of all graphs all possible orderings will appear, and 
a canonical ordering is not needed in physical contexts.
\subsection{Properties of Feynman graphs}
\begin{figure}
\begin{subfigure}[t]{0.22\textwidth}\centering
 \parbox[c][55pt][t]{50pt}{\centering
    \begin{fmffile}{disconnected_phi4_example_fg}
    \begin{fmfgraph}(50,50)
\end{fmfgraph}
    \end{fmffile}
    }\\
    $|\text{Aut}_\text{lf}| = 2$\\
    $|\text{Aut}_\text{nlf}| = 8$
    \caption{A disconnected graph with a self-loop (i.e. a tadpole graph).}
    \label{subfig:disconnected_with_selfloop}
\end{subfigure}
~~
\begin{subfigure}[t]{0.22\textwidth}\centering
 \parbox[c][55pt][t]{50pt}{\centering
    \begin{fmffile}{1pi_phi3_example_fg}
    \begin{fmfgraph}(50,50)
\end{fmfgraph}
    \end{fmffile}
    }\\
    $|\text{Aut}_\text{lf}| = 4$\\
    $|\text{Aut}_\text{nlf}| = 32$
    \caption{A 1PI, but not $2$-vertex connected and not tadpole graph.}
    \label{subfig:1pi_wo_vtx2connted}
\end{subfigure}\\
\begin{subfigure}[t]{0.22\textwidth}\centering
 \parbox[c][55pt][t]{50pt}{\centering
    \begin{fmffile}{2vtx_connected_example_fg}
    \begin{fmfgraph}(50,50)
\end{fmfgraph}
    \end{fmffile}
    }\\
    $|\text{Aut}_\text{lf}| = 4$\\
    $|\text{Aut}_\text{nlf}| = 16$
    \caption{A 2-vertex-connected graph.}
    \label{subfig:vtx_2_connected_graph}
\end{subfigure}
~~
\begin{subfigure}[t]{0.22\textwidth}\centering
 \parbox[c][55pt][t]{50pt}{\centering
    \begin{fmffile}{tadpole_phi3_example_fg}
    \begin{fmfgraph}(50,50)
\end{fmfgraph}
    \end{fmffile}
    }\\
    $|\text{Aut}_\text{lf}| = 12$\\
    $|\text{Aut}_\text{nlf}| = 24$
    \caption{A 1PI tadpole graph without self-loops.}
    \label{subfig:tadpole_wo_selfloop}
\end{subfigure}
\caption{Examples of Feynman graphs, which fulfill certain properties, with the orders of their automorphism groups.}
\label{fig:fg_props_example}
\end{figure}
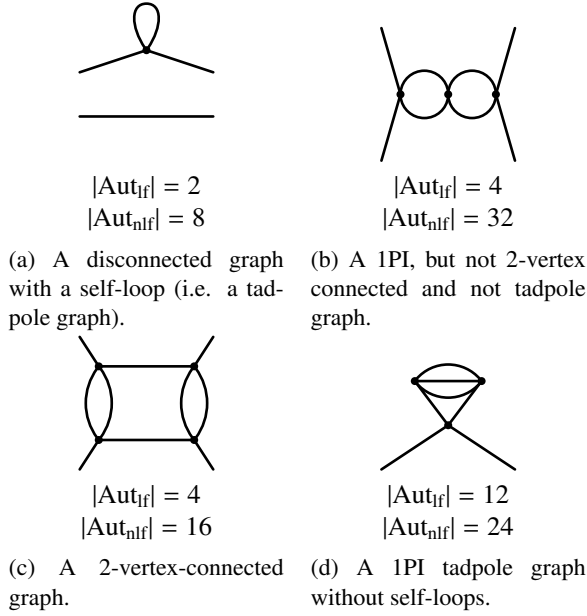
\paragraph{1PI (one particle irreducibleness)}
A graph, which is still connected if any internal edge is removed, is called a 1PI graph. 

Figure \ref{fig:fg_props_example} (\subref{subfig:1pi_wo_vtx2connted}), \ref{fig:fg_props_example} (\subref{subfig:vtx_2_connected_graph}) and \ref{fig:fg_props_example} (\subref{subfig:tadpole_wo_selfloop}) show examples of 1PI graphs. 
In figure \ref{fig:fg_props_example} (\subref{subfig:disconnected_with_selfloop}) a disconnected graph is depicted, which is thereby not 1PI.
\paragraph{$2$-vertex-connectivity}
A graph which has no self-loops and is still connected if any vertex $v$ is removed 
together with the edges incident to $v$, is called $2$-vertex-connected.
In figure \ref{fig:fg_props_example} (\subref{subfig:vtx_2_connected_graph}) a $2$-vertex-connected graph is shown. The other graphs in figure \ref{fig:fg_props_example} are not $2$-vertex-connected.
\paragraph{Tadpole graphs}
A not $2$-vertex-connected graph $\Gamma$, which either has self-loops or can be split into 
connected components upon removal of any vertex, such that one of the new connected 
components is not connected to an external vertex, is called a tadpole graph.
For graphs without external legs, this notion coincides with $2$-vertex connectivity.
Figure \ref{fig:fg_props_example} (\subref{subfig:disconnected_with_selfloop}) depicts a disconnected tadpole graph and 
figure \ref{fig:fg_props_example} (\subref{subfig:tadpole_wo_selfloop}) shows an example of a tadpole graph without self-loops.
\subsection{Contractions of subgraphs}
\label{subsec:contractions}
An important notion in connection to the Hopf algebra of Feynman graphs 
is the contraction of subgraphs. 
If $\gamma \subseteq \Gamma$ is a union of disjoint one-particle irreducible subgraphs $\gamma0\cup\gamma_i$, the contraction, $\Gamma/\gamma$, 
of $\gamma$ in $\Gamma$ is obtained by removing all edges 
$e \in \gamma^{[1]}_\text{int} \subseteq \Gamma^{[1]}_\text{int}$ 
of $\Gamma$ and by merging all vertices connected to these edges to 
one new vertex for each subgraph $\gamma_i$. If the new vertex is of valency $2$, it is also 
removed and the two edges incident to it are joined to one edge. This can be modified in need of separate mass and kinetic renormalization for massive propagators.

Contractions are commutative operations if disjoint subgraphs are contracted. 
For a graph $\Gamma$ with $\gamma_1, \gamma_2 \subseteq \Gamma$ and 
$\gamma_1 \cap \gamma_2 = \emptyset$:
\begin{align}
    \left(\Gamma / \gamma_1\right) / \gamma_2 = \left(\Gamma / \gamma_2\right) / \gamma_1
    = \Gamma / \left(\gamma_1 \cup \gamma_2\right).
\end{align}
Furthermore, contractions can be ``canceled'', for instance 
for $\delta \subseteq \gamma \subseteq \Gamma$ 
\begin{align}
\left(\Gamma/\delta\right) / \left(\gamma/\delta\right) = \Gamma / \gamma.
\end{align}

\subsection{Residues of graphs}
Using the notion of contractions, the map $\text{res}$ can be extended to act on connected Feynman graphs $\Gamma$. 
The maximal contraction, $\Gamma/\Gamma$, is formed, which contracts 
all internal edges of $\Gamma$, such that only 
one internal vertex or an external edge is left. 
$\text{res}(\Gamma)$ denotes the vertex or edge type of the left over vertex or edge.
This is called the residue type of $\Gamma$. 
\paragraph{Examples}
\begin{align*}
&
\text{res} \left( 
        \parbox[c][25pt][t]{50pt}{\centering
   \begin{fmffile}{qed_2l_1pi_prop_omega_example}
    \begin{fmfgraph}(50,25)
\end{fmfgraph}
    \end{fmffile}}
        \right) = \parbox[c][10pt][t]{12.5pt}{\centering
        \begin{fmffile}{qed_edge2}
        \begin{fmfgraph}(10,10)
\end{fmfgraph}
        \end{fmffile}} & 
&\text{res} \left(
    \parbox[c][25pt][t]{25pt}{\centering
   \begin{fmffile}{qed_1pi_1l_omega_example1}
    \begin{fmfgraph}(25,25)
\end{fmfgraph}
    \end{fmffile}}
\right) = \parbox[c][15pt][t]{17.0pt}{\centering
        \begin{fmffile}{qed3_vtx}
        \begin{fmfgraph}(15,15)
\end{fmfgraph}
        \end{fmffile}} \\
&\text{res} \left(
        \parbox[c][25pt][t]{50pt}{\centering
   \begin{fmffile}{phi3_1l_1pi_omega_example}
    \begin{fmfgraph}(50,25)
\end{fmfgraph}
    \end{fmffile}}
      \right)= 
 \parbox[c][10pt][t]{12.5pt}{\centering
        \begin{fmffile}{phi4_edge_example}
        \begin{fmfgraph}(10,10)
\end{fmfgraph}
        \end{fmffile}} &
&\text{res} \left(
\parbox[c][25pt][t]{10pt}{\centering
    \begin{fmffile}{phi4_1loop_omega_example}
    \begin{fmfgraph}(10,25)
\end{fmfgraph}
    \end{fmffile}}
\right) = \parbox[c][10pt][t]{12.5pt}{\centering
        \begin{fmffile}{phi4_vtx}
        \begin{fmfgraph}(10,10)
\end{fmfgraph}
        \end{fmffile}}
\end{align*}

\subsection{Weight $\omega_D$ of a Feynman graph}
\paragraph{Vertex and edge weights}
The Feynman rules give rise to a weight for every vertex and edge type:
\begin{align}
\omega : \mathcal{R}_V \cup \mathcal{R}_E \rightarrow \mathbb{Z}.
\end{align}
This weight corresponds to the negative power of the momenta in the associated Feynman rule.

In QED for example, a weight of $2$ is assigned to photon edges,
$\omega\left(\parbox[c][10pt][t]{12.5pt}{\centering
        \begin{fmffile}{qed_edge2}
        \begin{fmfgraph}(10,10)
\end{fmfgraph}
        \end{fmffile}} \right)=2$,
and $1$ is assigned to fermion edges, 
$\omega\left(\parbox[c][10pt][t]{12.5pt}{\centering
        \begin{fmffile}{qed_edge1}
        \begin{fmfgraph}(10,10)
\end{fmfgraph}
        \end{fmffile}}\right)=1$.
Because QED vertices do not depend on any momenta, their weight is $0$, 
$\omega\left(
\parbox[c][15pt][t]{17.0pt}{\centering
        \begin{fmffile}{qed3_vtx}
        \begin{fmfgraph}(15,15)
\end{fmfgraph}
        \end{fmffile}}\right)=0$.

In $\varphi^k$-theory all edges are assigned the weight $2$, $\omega\left(
\parbox[c][10pt][t]{12.5pt}{\centering
        \begin{fmffile}{phi3_edge}
        \begin{fmfgraph}(10,10)
\end{fmfgraph}
        \end{fmffile}}\right)=2$ and the vertices have weight $0$.

\paragraph{The map $\omega_D$ on Feynman graphs}
Using the map $\omega$ for given Feynman rules to assign 
a weight to every vertex and edge type, an additional map $\omega_D$ can 
be defined, giving a weight to a Feynman graph:
\begin{mydef}
\label{def_omega_D}
\begin{align}
    \omega_D\left(\Gamma\right) &:= \sum \limits_{v\in \Gamma^{[0]}_\text{int}} \omega(\text{\normalfont{res}}(v)) + \sum \limits_{e\in \Gamma^{[1]}_\text{int}} \omega(\text{\normalfont{res}}(e))  - D h_1 \left(\Gamma\right)
\end{align}
where $h_1(\Gamma)$ denotes the loop number of $\Gamma$ (the first Betti number)  and $D$ is 
a parameter that will be associated with the spacetime dimension.
Neglecting possible infrared divergences, the value of $\omega_D$ coincides
with the degree of divergence of the integral associated to the graph 
in a $D$-dimensional quantum field theory.
To the empty graph the weight $0$ is assigned: $\omega_D\left(\emptyset\right) = 0$. 
A 1PI graph $\Gamma$ with $\omega_D(\Gamma) \leq 0$ is called superficially divergent in $D$ dimensions.
\end{mydef}
\paragraph{Examples}
\begin{gather*}
\omega_4 \left( 
        \parbox[c][25pt][t]{50pt}{\centering
   \begin{fmffile}{qed_2l_1pi_prop_omega_example}
    \begin{fmfgraph}(50,25)
\end{fmfgraph}
    \end{fmffile}}
        \right) = 
\omega_6 \left( 
        \parbox[c][25pt][t]{50pt}{\centering
   \begin{fmffile}{phi3_2l_1pi_prop_omega_example}
    \begin{fmfgraph}(50,25)
\end{fmfgraph}
    \end{fmffile}}
    \right) = -2 \\
\omega_4 \left(
    \parbox[c][25pt][t]{25pt}{\centering
   \begin{fmffile}{qed_1pi_1l_omega_example}
    \begin{fmfgraph}(25,25)
\end{fmfgraph}
    \end{fmffile}}
\right) = 
\omega_4 \left(
\parbox[c][25pt][t]{25pt}{\centering
    \begin{fmffile}{qed_1loop_omega_example2}
    \begin{fmfgraph}(25,25)
\end{fmfgraph}
    \end{fmffile}}
\right) = 
\omega_4 \left(
\parbox[c][25pt][t]{10pt}{\centering
    \begin{fmffile}{phi4_1loop_omega_example}
    \begin{fmfgraph}(10,25)
\end{fmfgraph}
    \end{fmffile}}
\right) = 
\omega_6 \left(
\parbox[c][25pt][t]{25pt}{\centering
    \begin{fmffile}{phi3_1loop_omega_example}
    \begin{fmfgraph}(25,25)
\end{fmfgraph}
    \end{fmffile}}
\right) =0\\
\omega_4 \left(
\parbox[c][25pt][t]{25pt}{\centering
    \begin{fmffile}{qed_1loop_omega_example3}
    \begin{fmfgraph}(25,25)
\end{fmfgraph}
    \end{fmffile}}
\right) = 
\omega_4 \left(
\parbox[c][25pt][t]{25pt}{\centering
    \begin{fmffile}{phi4_1loop_example2}
    \begin{fmfgraph}(25,25)
\end{fmfgraph}
    \end{fmffile}
    }
\right) = 
\omega_6 \left(
\parbox[c][25pt][t]{25pt}{\centering
    \begin{fmffile}{phi3_1loop_omega_example2}
    \begin{fmfgraph}(25,25)
\end{fmfgraph}
    \end{fmffile}}
\right) =2
\end{gather*}

\section{The Hopf algebra of Feynman graphs}
\label{chap:hopf_algebra}
\subsection{Definition}
\label{sec:def_hopf_algebra}
In this section a short definition of the Hopf algebra of Feynman graphs is given. 
For a detailed definition consult \cite{connes2000renormalization,hopf_feynman,Manchon}.

Let $\mathcal{T}$ be the set of all non-isomorphic 1PI graphs of a given 
quantum field theory, including the empty graph. $\mathcal{F}$ is the free commutative monoid 
generated by the elements in $\mathcal{T}$. The empty graph is associated with the neutral element 
$\mathbb{I} \in \mathcal{F}$.

Following \cite{hopf_feynman}, the Hopf algebra $\mathcal{H}_D$ is a vector space spanned by the 
elements of $\mathcal{F}$ with the multiplication $m$ on $\mathcal{F}$ extended to 
linear combinations of products of graphs.
Additionally, $\mathcal{H}_D$ is equipped with a linear map called the coproduct
$\Delta_D : \mathcal{H}_D \rightarrow \mathcal{H}_D \otimes \mathcal{H}_D$, which depends on the dimension $D$.
For 1PI graphs $\Gamma\in\mathcal{T}\subset\mathcal{H}_D$ it is defined as
\begin{mydef}
\label{def_cop}
\begin{align}
\label{eqn:def_cop}
    &\Delta_D \Gamma := 
    \sum \limits_{ \substack{ \gamma \unlhd \Gamma} }
        \gamma \otimes \Gamma/\gamma& &:& &\mathcal{T} \rightarrow \mathcal{H}_D \otimes \mathcal{H}_D
\end{align}
where
\begin{gather}
\begin{gathered}
\label{eqn:relation_unlhd}
\gamma \unlhd \Gamma \\ 
\Leftrightarrow \\ 
 \gamma \in 
\left\{ \delta\subseteq \Gamma \left| \delta = 
\bigcup \limits_i \delta_i \text{, \normalfont{s.t.} }
\delta_i \in \mathcal{T} \right. 
\text{ \normalfont{and} } \omega_D(\delta_i) \le 0 \right\} 
\end{gathered}
\end{gather}
denotes the membership of $\gamma$ in the set of subgraphs of $\Gamma$, 
whose connected components are superficially divergent 1PI graphs. 
Disconnected graphs $\gamma=\bigcup \limits_i \gamma_i$ are identified 
with the product $\left(\prod \limits_i \gamma_i\right) \in \mathcal{F} \subset \mathcal{H}_D$.
The cograph $\Gamma/\Gamma$ and the empty graph $\gamma = \emptyset$ in the sum in 
\eqref{eqn:def_cop} are identified with $\mathbb{I} \in \mathcal{H}_D$.
\end{mydef}
Note, that while $\Gamma/\gamma$, the so-called cographs, in the sum in \eqref{eqn:def_cop} can inherit a numbering of the external edges 
from $\Gamma$, it is not possible to assign a numbering to the external legs of the subgraphs $\gamma$. 
That means, if the corresponding spaces of leg-fixed (lf) and non-leg-fixed (nlf) 1PI graphs, $\mathcal{T}^\text{lf}$, 
$\mathcal{T}^\text{nlf}$, ${\mathcal{H}_D}^\text{lf}$, ${\mathcal{H}_D}^\text{nlf}$ are 
distinguished, $\Delta_D$ maps the spaces as follows:
\begin{align*}
&\Delta_D& &:& &\mathcal{T}^\text{nlf}& &\rightarrow& &\mathcal{H}_D^\text{nlf}& &\otimes& &\mathcal{H}_D^\text{nlf} \\
&\Delta_D& &:& &\mathcal{T}^\text{lf}&     &\rightarrow& &\mathcal{H}_D^\text{nlf}& &\otimes& &\mathcal{H}_D^\text{lf}.
\end{align*}
Strictly speaking, only a coaction not a coproduct is obtained in the lf case. 
A coproduct can be obtained by defining a map $\mathcal{T}^\text{nlf} \rightarrow \mathcal{T}^\text{lf}$ 
which maps a nlf graph to the sum of all lf graphs which correspond to the nlf graph if 
the ordering of the external edges is ignored. Therefore, $\Delta_D$ can still
be promoted to a coproduct in the leg-fixed case.
Because here, the subgraphs themselves and not the algebra elements associated with them are of 
most interest, this ambiguity will not play a major role in this work. 
For most of the properties of $\mathcal{H}_D$, it is not relevant whether the 
set of leg-fixed or the one of non-leg-fixed 1PI graphs is taken as generators.
An exception is section \ref{sec:1pisum}, which will only cope with 
non-leg-fixed graphs for simplicity. 

The map $\Delta_D$ is an algebra morphism $\Delta_D \left(\Gamma_1 \Gamma_2\right) = 
\left( \Delta_D \Gamma_1 \right) \left( \Delta_D \Gamma_2 \right)$
In accordance to this, $\Delta_D \mathbb{I} := \mathbb{I} \otimes \mathbb{I}$.
Furthermore, $\Delta_D$ is a coassociative map: 
\begin{align}
    \left( \Delta_D \otimes \text{\normalfont{id}} \right) \Delta_D = 
    \left( \text{\normalfont{id}} \otimes \Delta_D \right) \Delta_D.
\end{align}

Another important notion is the reduced coproduct:
\begin{mydef}
\label{def:red_cop}
The reduced coproduct $\widetilde{\Delta}_D$ is defined as
\begin{align}
\label{eqn:def_red_cop}
    \widetilde{\Delta}_D := \Delta_D - \text{id} \otimes \mathbb{I} - \mathbb{I} \otimes \text{id}
    & &:& &\mathcal{H}_D \rightarrow \mathcal{H}_D \otimes \mathcal{H}_D.
\end{align}
\end{mydef}
It gives rise to the space of primitive elements of $\mathcal{H}_D$:
\begin{mydef}
\label{def:primitive}
\begin{align}
    \text{Prim}\left(\mathcal{H}_D\right) := \text{ker } \widetilde{\Delta}.
\end{align}
\end{mydef}

Furthermore, $\mathcal{H}_D$ is connected and graded by the loop number $h_1(\Gamma)$.

\subsection{Sum of the coproducts of all 1PI graphs}
\label{sec:1pisum}
Having established the basic properties 
of the Hopf algebra of Feynman graphs, an 
additional identity, suitable to test the 
functionality of \textbf{feyncop}, will be 
derived in this section. 
To do so, some further properties of Feynman graphs 
must be introduced. For simplicity, this section will only 
take non-leg-fixed graphs into account. 

\subsubsection{Numbers of vertices, edges and connected components of certain types}
For every residue type $r \in \mathcal{R}_V \cup \mathcal{R}_E$ and graph $\Gamma \in \mathcal{T}$,
\begin{align}
    m_r( \Gamma ) := \left| \left\{ \left. t \in \Gamma^{[0]}_\text{int} \cup \Gamma^{[1]}_\text{int} \right| \text{res}(t) = r \right\} \right|
\end{align}
gives the number of internal vertices or edges in $\Gamma$ of type $r$.

Suppose, $\gamma$ is a product of graphs,
$\gamma = \left(\prod \limits_i^n \gamma_i\right)\in \mathcal{F}$, then
\begin{align}
    n_r( \gamma ) := \left| \left\{ \left. \gamma_i\in\left\{ \gamma_1, \ldots, \gamma_n \right\} \right| \text{res}(\gamma_i) = r \right\} \right|
\end{align}
denotes the number of factors of $\gamma$ with residue $r$.

Let $\mathcal{R}_{E}'$ be the set of edge residue types taking 
different orientations into account. That means that for 
each oriented edge type in $r_E \in \mathcal{R}_E$ there are two edge types
in $\mathcal{R}_{E}'$ each corresponding to one of the possible orientations. 
Edge types without orientation appear once in $\mathcal{R}_{E}'$ as in $\mathcal{R}_{E}$.
For QED for instance this means ${\mathcal{R}'}_{E}^{\text{QED}} = \left\{
\parbox[c][10pt][t]{12.5pt}{\centering
        \begin{fmffile}{qed_edge1}
        \begin{fmfgraph}(10,10)
\end{fmfgraph}
        \end{fmffile}},
\parbox[c][10pt][t]{12.5pt}{\centering
        \begin{fmffile}{qed_edge1_dir}
        \begin{fmfgraph}(10,10)
\end{fmfgraph}
        \end{fmffile}},
\parbox[c][10pt][t]{12.5pt}{\centering
        \begin{fmffile}{qed_edge2}
        \begin{fmfgraph}(10,10)
\end{fmfgraph}
        \end{fmffile}}
\right\}$.

Given some edge type $r_E \in \mathcal{R}_{E}'$ and a vertex type $r_V \in \mathcal{R}_{V}$, 
\begin{align}
N_{r_E} \left(r_V\right)
\end{align}
is the number of edges of type $r_E$ incident to vertices of type $r_V$ in the given orientation.
For QED the values are,
\begin{align}
    N_{\parbox[c][10pt][t]{12.5pt}{\centering
        \begin{fmffile}{qed_edge1}
        \begin{fmfgraph}(10,10)
\end{fmfgraph}
        \end{fmffile}}} \left(\parbox[c][15pt][t]{17.0pt}{\centering
        \begin{fmffile}{qed3_vtx}
        \begin{fmfgraph}(15,15)
\end{fmfgraph}
        \end{fmffile}}\right)
= 
    N_{\parbox[c][10pt][t]{12.5pt}{\centering
        \begin{fmffile}{qed_edge1_dir}
        \begin{fmfgraph}(10,10)
\end{fmfgraph}
        \end{fmffile}}} \left(\parbox[c][15pt][t]{17.0pt}{\centering
        \begin{fmffile}{qed3_vtx}
        \begin{fmfgraph}(15,15)
\end{fmfgraph}
        \end{fmffile}}\right)
= 
    N_{\parbox[c][10pt][t]{12.5pt}{\centering
        \begin{fmffile}{qed_edge2}
        \begin{fmfgraph}(10,10)
\end{fmfgraph}
        \end{fmffile}}} \left(\parbox[c][15pt][t]{17.0pt}{\centering
        \begin{fmffile}{qed3_vtx}
        \begin{fmfgraph}(15,15)
\end{fmfgraph}
        \end{fmffile}}\right) = 1.
\end{align}
On the other hand for $\varphi^3$-theory, the only relevant value is
\begin{align}
    N_{\parbox[c][10pt][t]{12.5pt}{\centering
        \begin{fmffile}{phi3_edge}
        \begin{fmfgraph}(10,10)
\end{fmfgraph}
        \end{fmffile}}} \left(\parbox[c][10pt][t]{12.5pt}{\centering
        \begin{fmffile}{phi3_vtx}
        \begin{fmfgraph}(10,10)
\end{fmfgraph}
        \end{fmffile}}\right) = 3.
\end{align}

\subsubsection{Permuting external legs}
Let $\text{Perm}_\text{ext}(\Gamma)$ for $\Gamma \in \mathcal{T}$ be
the group of permutations of external vertices of $\Gamma$, 
preserving the external vertex types.
By permuting the external vertices of $\Gamma$ 
new graphs can be obtained. In the case of 
non-leg-fixed graphs, these new graphs will be
isomorphic to the original one. 
For $\Gamma \in \mathcal{T}$ the 
number of different permutations is given as 
\begin{align}
\prod \limits_{r_E \in \mathcal{R}_{E}'} N_{r_E}\left( \text{res}(\Gamma) \right)!.
\end{align}
For a product of graphs $\gamma = \left(\prod \limits_i \gamma_i \right) \in \mathcal{F}$ with some numbering of the factors chosen, 
$\text{Perm}_\text{ext}(\gamma)$ is defined as the product group:
\begin{align}
\text{Perm}_\text{ext}(\gamma) = \text{Perm}_\text{ext}(\gamma_1) \times \text{Perm}_\text{ext}(\gamma_2) \times \ldots
\end{align}
The cardinality of this group is 
\begin{align}
\left| \text{Perm}_\text{ext}(\gamma) \right| = 
\prod \limits_{r \in \mathcal{R}_{V} \cup \mathcal{R}_E} 
\left(
\prod \limits_{r_E \in \mathcal{R}_{E}'} N_{r_E}( r_V )!
\right) ^ {n_{r}(\gamma)}.
\end{align}

\subsubsection{Insertions}
Inserting a graph into another can be interpreted as 
an inverse to the operation of contracting a subgraph.
Suppose $\gamma,\Gamma \in \mathcal{T}$ and $\gamma$ should be inserted into $\Gamma$.
If $\gamma$ has an edge type residue, 
it can be glued in an internal edge $e\in\Gamma^{[1]}_\text{int}$ with 
$\text{res}(e) = \text{res}(\gamma)$ and if 
it has a vertex type residue, it can be inserted to 
replace an internal vertex $v\in\Gamma^{[0]}_\text{int}$ with $\text{res}(v) = \text{res}(\gamma)$.
This can be extended to products of graphs $\gamma = \left(\prod \limits_i \gamma_i\right) \in \mathcal{F}$ 
being inserted into $\Gamma\in\mathcal{T}$, by inserting 
all the graphs $\gamma_i$ subsequently into different 
vertices or edges of $\Gamma$.
Note that it is possible to insert multiple graphs of 
edge residue type
into the same edge of $\Gamma$, whereas it is not 
possible to insert more than one of vertex residue type
into the same vertex of $\Gamma$.

The set of insertions of $\gamma\in\mathcal{F}$ into $\Gamma\in\mathcal{T}$, 
yielding possibly isomorphic graphs, is denoted by $\mathcal{I}(\Gamma | \gamma )$.
The graph obtained, when $\gamma$ is inserted into $\Gamma$ using an 
insertion $i \in\mathcal{I}(\Gamma | \gamma )$, is denoted by
$\Gamma \circ_i \gamma \in \mathcal{T}$.

The cardinality of $\mathcal{I}(\Gamma | \gamma )$ is, 
\begin{align}
\begin{split}
\label{eqn:num_insertions}
    \left| \mathcal{I}(\Gamma | \gamma ) \right| &= \left| \text{Perm}_\text{ext}(\gamma) \right|
    \prod \limits_{r_V \in \mathcal{R}_V} n_{r_V}(\gamma)! \binom {m_{r_V}(\Gamma)} {n_{r_V}(\gamma)} \times \\
&\times
    \prod \limits_{r_E \in \mathcal{R}_E} n_{r_E}(\gamma)! \binom {m_{r_E}(\Gamma) + n_{r_E}(\gamma) - 1} {n_{r_E}(\gamma)} ,
\end{split}
\end{align}
where the first term describes the freedom to insert $\gamma$ into  
$\Gamma$ in all permutations of $\gamma$'s
external legs, the second term represents the number of choices of 
insertion places for the factors of vertex residue type and
the last term stands for the number of choices to insert factors
of edge type into $\Gamma$. The factorials count the number of ways 
to choose a suitable order of the factors to insert.
The difference in the vertex and edge insertions results from the fact that 
edges can be used for multiple insertions and vertices only for one.
The special cases are declared as
\begin{align}
\label{eqn:1pisum_special}
    \left| \mathcal{I}(\Gamma | \mathbb{I} ) \right| &=  \left| \mathcal{I}(\mathbb{I} | \Gamma ) \right| = \left| \mathcal{I}(\mathbb{I} | \mathbb{I} ) \right| = 1 & & \forall \Gamma \in \mathcal{T} \\
\notag
\text{and } \left| \mathcal{I}(\mathbb{I} | \gamma ) \right|&= 0 & &\forall \gamma \in \mathcal{F}, \gamma \notin \mathcal{T}.
\end{align}

\paragraph{Examples}
Suppose 
$ \Gamma = \gamma = \parbox[c][20pt][t]{20pt}{\centering
        \begin{fmffile}{insertion_example1}
        \begin{fmfgraph}(20,20)
\end{fmfgraph}
        \end{fmffile}}
$
and $\gamma$ shall be inserted into $\Gamma$. 
Dealing with a $\varphi^4$-theory only 
one vertex and one edge type need to be considered.
$\gamma$ is connected and of vertex residue type. Therefore, 
$n_{
\parbox[c][7pt][t]{9pt}{\centering
        \begin{fmffile}{phi4_vtx_insertion}
        \begin{fmfgraph}(7,7)
\end{fmfgraph}
        \end{fmffile}}
}(\gamma) = 1$ and 
$n_{
\parbox[c][7pt][t]{9pt}{\centering
        \begin{fmffile}{phi4_edge_insertion}
        \begin{fmfgraph}(7,7)
\end{fmfgraph}
        \end{fmffile}}
}(\gamma) = 0$. 
$\left| \text{Perm}_\text{ext} (\gamma) \right| = 4!=24$, because 
of the four external legs of $\gamma$ of similar type.

$\Gamma$ has two internal edges and two 
internal vertices. Consequently, $m_{
\parbox[c][7pt][t]{9pt}{\centering
        \begin{fmffile}{phi4_vtx_insertion}
        \begin{fmfgraph}(7,7)
\end{fmfgraph}
        \end{fmffile}}}(\Gamma) = 2$ and  
$m_{
\parbox[c][7pt][t]{9pt}{\centering
        \begin{fmffile}{phi4_edge_insertion}
        \begin{fmfgraph}(7,7)
\end{fmfgraph}
        \end{fmffile}}
}(\Gamma) = 2$.
The total number of insertions is then, according to formula \eqref{eqn:num_insertions},
\begin{align*}
\left| \mathcal{I}(\Gamma | \gamma ) \right| = 24 \cdot 1! \cdot \binom{2}{1} \cdot 0! \cdot \binom{2-1}{0} = 48.
\end{align*}

Only two non isomorphic graphs will be obtained if $\Gamma \circ_i \gamma$ is formed for all $i \in \mathcal{I}(\Gamma | \gamma )$. 
Considering the possible permutations of the external legs 
of $\gamma$, there will be $32$ insertions giving graphs 
of the same isomorphism class as
$\parbox[c][20pt][t]{20pt}{\centering
     \begin{fmffile}{insertion_example2}
        \begin{fmfgraph}(20,20)
\end{fmfgraph}
        \end{fmffile}}
$
and $16$ insertions will give graphs of the same isomorphism class as 
$
\parbox[c][20pt][t]{20pt}{\centering
     \begin{fmffile}{insertion_example3}
        \begin{fmfgraph}(20,20)
\end{fmfgraph}
        \end{fmffile}}.
$

If on the other hand the product of graphs 
$
    \gamma' =
\left(\parbox[c][22pt][t]{22pt}{\centering
    \begin{fmffile}{insertion_example4}
    \begin{fmfgraph}(20,20)
\end{fmfgraph}
    \end{fmffile}
    }\right)^2
$
shall be inserted into $\Gamma$, the edge insertion places of 
$\Gamma$ need to be considered.
Clearly, $n_{
\parbox[c][7pt][t]{9pt}{\centering
        \begin{fmffile}{phi4_vtx_insertion}
        \begin{fmfgraph}(7,7)
\end{fmfgraph}
        \end{fmffile}}
}(\gamma') = 0$,
$n_{
\parbox[c][7pt][t]{9pt}{\centering
        \begin{fmffile}{phi4_edge_insertion}
        \begin{fmfgraph}(7,7)
\end{fmfgraph}
        \end{fmffile}}
}(\gamma') = 2$ and 
$\left| \text{Perm}_\text{ext} (\gamma') \right| = \left(2!\right)^2=4$, 
because both factors of $\gamma'$ have two external legs of the same 
type.
Plugging into formula \eqref{eqn:num_insertions} yields
\begin{align*}
\left| \mathcal{I}(\Gamma | \gamma' ) \right| = 4 \cdot 0! \cdot \binom{2}{0} 
\cdot 2! \cdot \binom{2+2-1}{2} = 24.
\end{align*}
Forming $\Gamma \circ_i \gamma'$ for all $i\in\mathcal{I}( \Gamma | \gamma' )$ 
will yield $16$ graphs isomorphic to
$
    \parbox[c][20pt][t]{20pt}{\centering
    \begin{fmffile}{insertion_example5}
    \begin{fmfgraph}(20,20)
\end{fmfgraph}
    \end{fmffile}}
$
and $8$ graphs isomorphic to
$
    \parbox[c][20pt][t]{20pt}{\centering
    \begin{fmffile}{insertion_example6}
    \begin{fmfgraph}(20,20)
\end{fmfgraph}
    \end{fmffile}}.
$

\subsubsection{Automorphism group of products of graphs}
Guided by the association of products of graphs with disconnected graphs the 
notion of the automorphism group can be extended.
The extension to products of non-isomorphic graphs is trivial, because 
all factors can be distinguished:
\begin{align}
\begin{split}
\text{Aut}\left( \prod \limits_i \gamma_i \right) &= 
\text{Aut}\left( \gamma_1 \right) \times \text{Aut}\left( \gamma_2 \right) \times \ldots \\
 & \text{ with } \gamma_i \in \mathcal{T} \text{ and } \gamma_i \not\simeq \gamma_j ~~ \forall i,j 
\end{split}\\
\begin{split}
\Rightarrow \left|\text{Aut}\left( \prod \limits_i \gamma_i \right) \right| &= 
\prod \limits_i \left| \text{Aut}(\gamma_i) \right|  \\
 &\text{ with } \gamma_i \in \mathcal{T} \text{ and }\gamma_i \not\simeq \gamma_j ~~ \forall i,j
\end{split}
\end{align}
If the product contains more than one graph of the same isomorphism class, additional symmetry generators are obtained given by the permutations of the 
graphs of the same isomorphism class. Therefore, with given multiplicities $n_i$ for 
every isomorphism class the cardinality of the automorphism group is:
\begin{align}
\begin{split}
\left| \text{Aut}\left( \prod \limits_i \left( \gamma_i \right)^{n_i} \right) \right| &= 
\prod \limits_i \left( n_i! \left| \text{Aut}\left( \gamma_i \right) \right|^{n_i} \right) \\
 &\text{ with } \gamma_i \in \mathcal{T} \text{ and }\gamma_i \not\simeq \gamma_j ~~ \forall i,j.
\end{split}
\end{align}

\subsubsection{Sum formula for 1PI graphs}
\label{sec:finally_sum_formula}
Making use of the preceding definitions and properties of 
Feynman graphs an identity of 1PI graphs sums 
involving the coproduct is proved. 
The statement is also proved in \cite{suijlekom2007ren}. 
Here, a different argument is laid out based on the following lemma,
a variant of a theorem in \cite{connes2000renormalization}:
\begin{mylemma}
\label{lemma:ab}
For three given graphs $\gamma \in \mathcal{F}$, $\Gamma \in \mathcal{T}$ and $\widetilde\Gamma \in \mathcal{T}$, the set of 
pairs $(j_1,j_2)$ of an embedding $j_1$ of $\gamma$ into $\Gamma$ and an isomorphism $j_2$ between $\Gamma/j_1(\gamma)$ and $\widetilde\Gamma$, 
\begin{align}
    A(\gamma, \widetilde\Gamma, \Gamma) = 
    \left\{ \left(j_1, j_2 \right) \left| j_1 : \gamma \rightarrow \Gamma \text{ and }
    j_2 : \Gamma/j_1(\gamma) \rightarrow \widetilde\Gamma \right. \right\},
\end{align}
has the same cardinality as the set of all pairs $(i,j)$ of insertions $i \in \mathcal{I}(\widetilde\Gamma | \gamma)$ 
and isomorphisms $j$ between $\widetilde\Gamma \circ_i \gamma$ and $\Gamma$, 
\begin{align}
    B(\gamma, \widetilde\Gamma, \Gamma) = 
    \left\{ \left(i, j \right) \left| i \in \mathcal{I}(\widetilde\Gamma | \gamma) \text{ and }
    j : \widetilde\Gamma \circ_i \gamma \rightarrow \Gamma \right. \right\}.
\end{align}
\end{mylemma}
\begin{proof}
A bijection between the sets $A(\gamma, \widetilde\Gamma, \Gamma)$ 
and $B(\gamma, \widetilde\Gamma, \Gamma)$ is constructed. 

Given is a pair $(j_1, j_2)\in A(\gamma, \widetilde\Gamma, \Gamma)$. 
$j_1$ gives rise to a subgraph $j_1(\gamma) \subseteq \Gamma$.
This subgraph can be contracted to yield an insertion $i'\in\mathcal{I}\left(\Gamma/j_1(\gamma)|j_1(\gamma)\right)$.
Using $j_1$ and $j_2$ an insertion $i\in  \mathcal{I}\left(\widetilde\Gamma|\gamma\right)$ can be won. 
Clearly, there is an isomorphism $j:\widetilde\Gamma \circ_i \gamma \rightarrow \Gamma$, induced by $j_1$ and $j_2$
because $\left(\Gamma/j_1(\gamma)\right) \circ_{i'} j_1(\Gamma) = \Gamma$. 

This construction is reversible, because with these $(i,j)\in B(\gamma, \widetilde\Gamma, \Gamma)$, the isomorphism
$j : \widetilde\Gamma \circ_i \gamma \rightarrow \Gamma$ can be used to reconstruct $j_1$ 
by restricting it on $\gamma \subseteq \widetilde\Gamma \circ_i \gamma$. 
Contracting $\Gamma$ to $\Gamma/j_1(\gamma)$ and applying $j^{-1}$ to get 
$j^{-1}(\Gamma)/j^{-1}\left(j_1(\gamma)\right)$ gives $\widetilde\Gamma$. Therefore, $j_2$ is 
also retrieved.
This procedure is defined for all $(i,j)\in B(\gamma, \widetilde\Gamma, \Gamma)$ and establishes that 
this is a one-to-one.
\end{proof}
\begin{mycorollary}
\label{coro:pre_sum_formula}
For $\gamma \in \mathcal{F}$ and $\widetilde\Gamma \in \mathcal{T}$, 
\begin{align}
\label{eqn:formula_corollary}
 \frac{
 \left| \mathcal{I}(\widetilde\Gamma | \gamma ) \right| }{
 \left| \text{\normalfont{Aut}}(\gamma) \right| \left| \text{\normalfont{Aut}}(\widetilde\Gamma) \right|} = 
 \sum \limits_{\Gamma \in \mathcal{T}} \frac{ \left| \left\{ \gamma' \subseteq \Gamma \left| \gamma' \simeq \gamma \text{ \normalfont{and} } 
 \Gamma / \gamma' \simeq \widetilde\Gamma \right. \right\} 
 \right|}{\left| \text{\normalfont{Aut}}(\Gamma) \right| } .
\end{align}
\end{mycorollary}
\begin{proof}
Because the total number of automorphisms of a graph $\gamma$ is given by $\left| \text{Aut}(\gamma) \right|$
\begin{align}
  \begin{gathered}
    |A(\gamma, \widetilde\Gamma, \Gamma)| = \\ 
    \left| \text{Aut}(\gamma) \right| \left|  \text{Aut}(\widetilde\Gamma) \right| \left| \left\{ \gamma' \subseteq \Gamma \left| \gamma' \simeq \gamma \text{ and } 
 \Gamma / \gamma' \simeq \widetilde\Gamma \right. \right\} \right| \\
 \end{gathered}
\intertext{and}
    \begin{gathered}
    |B(\gamma, \widetilde\Gamma, \Gamma)| = \\ 
    \left| \text{Aut}(\Gamma) \right|
    \left| \left\{ 
    i \left| i \in \mathcal{I}(\widetilde\Gamma|\gamma) \text{ and } 
    \widetilde\Gamma \circ_i \gamma \simeq \Gamma \right.
    \right\} \right|.
  \end{gathered}
\end{align}
Using
\begin{align}
\sum \limits_{\Gamma \in \mathcal{T}} \left| \left\{ 
    i \left| i \in \mathcal{I}(\widetilde\Gamma|\gamma) \text{ and } 
    \widetilde\Gamma \circ_i \gamma \simeq \Gamma \right.
    \right\} \right| = \left| \mathcal{I}( \widetilde\Gamma|\gamma ) \right|,
\end{align}
and plugging in $\frac{|B(\gamma, \widetilde\Gamma, \Gamma)|}{\left| \text{Aut}(\Gamma)\right| }$ yields
\begin{gather}
\begin{gathered}
    \left| \mathcal{I}( \widetilde\Gamma|\gamma ) \right| = \sum \limits_{\Gamma \in \mathcal{T} } \frac{\left| B(\gamma, \widetilde\Gamma, \Gamma) \right|}
    {\left| \text{Aut}(\Gamma)\right|} = \\
    = \sum \limits_{\Gamma \in \mathcal{T} } \frac{\left| \text{Aut}(\gamma)\right| \left| \text{Aut}(\widetilde\Gamma)\right|\left| \left\{ \gamma' \subseteq \Gamma \left| \gamma' \simeq \gamma \text{ and } 
 \Gamma / \gamma' \simeq \widetilde\Gamma \right. \right\} 
 \right|}
    {\left| \text{Aut}(\Gamma)\right|} ,
\end{gathered}
\end{gather}
which confirms \eqref{eqn:formula_corollary}.
\end{proof}

Using this corollary, it is possible to prove the theorem, which is the purpose of this section:
\begin{mytheorem}
\label{prop:sum_formula}
\begin{align}
\label{eqn:sum_formula}
\sum \limits_{ \Gamma \in \mathcal{T}} \frac{\Delta_D \Gamma}{\left| \text{\normalfont{Aut}}(\Gamma) \right|} = 
\sum \limits_{\substack{\gamma = \left(\prod \limits_i \gamma_i\right) \in \mathcal{F}\\ \omega_D(\gamma_i) \leq 0}} \sum \limits_{\widetilde\Gamma \in \mathcal{T}}
\frac{ \left| \mathcal{I}( \widetilde\Gamma|\gamma ) \right| }{ 
    \left| \text{\normalfont{Aut}}( \gamma ) \right| \left| \text{\normalfont{Aut}}(\widetilde\Gamma) \right|}
    \gamma \otimes \widetilde\Gamma.
\end{align}
\end{mytheorem}
\begin{proof}
The right hand side of the statement can be rewritten using 
corollary \ref{coro:pre_sum_formula}:
\begin{align}
&\begin{gathered}
\sum \limits_{ \Gamma \in \mathcal{T}} \frac{\Delta_D \Gamma}{\left| \text{\normalfont{Aut}}(\Gamma) \right|} = \\
=\sum \limits_{\Gamma \in \mathcal{T}} \sum \limits_{\substack{\gamma = \left(\prod \limits_i \gamma_i\right) \in \mathcal{F}\\ \omega_D(\gamma_i) \leq 0}} 
    \sum \limits_{\widetilde\Gamma \in \mathcal{T}}
\frac{ \left| \left\{ \gamma' \subseteq \Gamma \left| \gamma' \simeq \gamma \text{ and } 
 \Gamma / \gamma' \simeq \widetilde\Gamma \right. \right\} 
 \right|}{\left| \text{\normalfont{Aut}}(\Gamma) \right| }
 \gamma \otimes \widetilde\Gamma.
\end{gathered}
\end{align}
Interchanging sums is permissible, because the grading guaranties that the left hand side of the statement,
restricted on a certain loop number, will be a finite sum of tensor products.
Performing the two inner sums and taking definition of the relation $\unlhd$ in \eqref{eqn:relation_unlhd} into account, gives the result
\begin{align}
&\sum \limits_{ \Gamma \in \mathcal{T}} \frac{\Delta_D \Gamma}{\left| \text{\normalfont{Aut}}(\Gamma) \right|} 
= \sum \limits_{\Gamma \in \mathcal{T}} \sum \limits_{\substack{\gamma' \unlhd \Gamma}}
 \frac{1}{\left| \text{\normalfont{Aut}}(\Gamma) \right| }
 \gamma' \otimes \Gamma/\gamma',
\end{align}
where the inner sum coincides with the definition in equation \eqref{eqn:def_cop} of the coproduct on Feynman graphs.
\end{proof}

This formula gives a non-trivial check of the coproduct computation in \textbf{feyncop}. 
Identity \eqref{eqn:sum_formula} can be restricted to graphs with a certain residue $r$,
because the cographs carry the same residue as the original graph, and 
a certain loop number $L$ using the grading. 
Applying this restriction to both sides of Eq. \eqref{eqn:sum_formula} gives rise to
\begin{mycorollary}
\label{coro:sum_formula}
\begin{align}
  \begin{gathered}
\label{eqn:sum_different}
\sum \limits_{ \substack{ \Gamma \in \mathcal{T}\\ h_1(\Gamma) = L\\\text{\normalfont{res}}(\Gamma) = r}} \frac{\Delta_D \Gamma}{\left| \text{\normalfont{Aut}}(\Gamma) \right|} = \\
=\sum \limits_{ \substack{ l_1, l_2 \geq 0 \\ l_1 + l_2 = L} }
\sum \limits_{\substack{\gamma = \left(\prod \limits_i \gamma_i\right) \in \mathcal{F}\\ \omega_D(\gamma_i) \leq 0 \\h_1(\gamma) = l_1}} 
\sum \limits_{\substack{\widetilde\Gamma \in \mathcal{T}\\h_1(\widetilde\Gamma)=l_2\\\text{\normalfont{res}}(\widetilde\Gamma) = r } }
\frac{ \left| \mathcal{I}( \widetilde\Gamma|\gamma ) \right| }{ 
    \left| \text{\normalfont{Aut}}( \gamma ) \right| \left| \text{\normalfont{Aut}}(\widetilde\Gamma) \right|}
    \gamma \otimes \widetilde\Gamma.
  \end{gathered}
\end{align}
\end{mycorollary}
\textbf{feyncop} implements a unit test which uses this formula to validate the computation if all 
graphs of a given loop number and a certain residue type are given as input. 
This check was performed for QED, Yang-Mills, $\varphi^3$ and $\varphi^4$ graphs 
up to at least fourth loop order and the residue types in $\mathcal{R}_E \cup \mathcal{R}_V$.
Additional checks were also performed with graphs of residue types not in  $\mathcal{R}_E \cup \mathcal{R}_V$.

\paragraph{Example}
To illustrate this validation, which is
performed using corollary \ref{coro:sum_formula}, 
an example is given.

Consider the sum of all two loop, propagator residue type $\varphi^4$-graphs weighted by their symmetry factor:
\begin{align*}
X^{
\parbox[c][7pt][t]{9pt}{\centering
        \begin{fmffile}{phi4_edge_insertion}
        \begin{fmfgraph}(7,7)
\end{fmfgraph}
        \end{fmffile}}
}_2:= 
\sum \limits_{ \substack{ \Gamma \in \mathcal{T}\\ h_1(\Gamma) = 2\\\text{\normalfont{res}}(\Gamma) = 
\parbox[c][7pt][t]{9pt}{\centering
        \begin{fmffile}{phi4_edge_insertion}
        \begin{fmfgraph}(7,7)
\end{fmfgraph}
        \end{fmffile}}
}} \frac{\Gamma}{\left| \text{\normalfont{Aut}}(\Gamma) \right|} = 
\frac18 \parbox[c][20pt][t]{30pt}{\centering
     \begin{fmffile}{double_tadpole2_fullexample}
        \begin{fmfgraph}(20,20)
\end{fmfgraph}
        \end{fmffile}}
+ 
\frac1{12}
\parbox[c][20pt][t]{40pt}{\centering
     \begin{fmffile}{triple_eye2_fullexample}
        \begin{fmfgraph}(30,20)
\end{fmfgraph}
        \end{fmffile}}.
\end{align*}
The coproducts of both graphs in four dimensions are 
\begin{gather*}
\Delta_4 \left( 
\parbox[c][20pt][t]{20pt}{\centering
\begin{fmffile}{double_tadpole2_fullexample}
        \begin{fmfgraph}(20,20)
\end{fmfgraph}
        \end{fmffile}}
\right) = \\
=\mathbb{I} \otimes 
\parbox[c][20pt][t]{20pt}{\centering
\begin{fmffile}{double_tadpole2_fullexample}
        \begin{fmfgraph}(20,20)
\end{fmfgraph}
        \end{fmffile}}
+
\parbox[c][20pt][t]{20pt}{\centering
\begin{fmffile}{double_tadpole2_fullexample}
        \begin{fmfgraph}(20,20)
\end{fmfgraph}
        \end{fmffile}}
\otimes 
\mathbb{I}
+
\parbox[c][20pt][t]{20pt}{\centering
    \begin{fmffile}{tadpole_fullexample}
    \begin{fmfgraph}(20,20)
\end{fmfgraph}
    \end{fmffile}
    }
\otimes
\parbox[c][20pt][t]{20pt}{\centering
    \begin{fmffile}{tadpole_fullexample}
    \begin{fmfgraph}(20,20)
\end{fmfgraph}
    \end{fmffile}
    }
+ 
\parbox[c][20pt][t]{20pt}{\centering
        \begin{fmffile}{eye_fullexample}
        \begin{fmfgraph}(20,20)
\end{fmfgraph}
        \end{fmffile}}
\otimes
\parbox[c][20pt][t]{20pt}{\centering
    \begin{fmffile}{tadpole_fullexample}
    \begin{fmfgraph}(20,20)
\end{fmfgraph}
    \end{fmffile}
    }\\
\Delta_4 \left( 
\parbox[c][20pt][t]{30pt}{\centering
     \begin{fmffile}{triple_eye2_fullexample}
        \begin{fmfgraph}(30,20)
\end{fmfgraph}
        \end{fmffile}}
\right) = \\
=\mathbb{I} 
\otimes
\parbox[c][20pt][t]{30pt}{\centering
     \begin{fmffile}{triple_eye2_fullexample}
        \begin{fmfgraph}(30,20)
\end{fmfgraph}
        \end{fmffile}}
+
\parbox[c][20pt][t]{30pt}{\centering
     \begin{fmffile}{triple_eye2_fullexample}
        \begin{fmfgraph}(30,20)
\end{fmfgraph}
        \end{fmffile}}
\otimes
\mathbb{I}
+
3 
\parbox[c][20pt][t]{20pt}{\centering
        \begin{fmffile}{eye_fullexample}
        \begin{fmfgraph}(20,20)
\end{fmfgraph}
        \end{fmffile}}
\otimes
\parbox[c][20pt][t]{20pt}{\centering
    \begin{fmffile}{tadpole_fullexample}
    \begin{fmfgraph}(20,20)
\end{fmfgraph}
    \end{fmffile}
    }.
\end{gather*}
Therefore, applying the coproduct according to definition \ref{def_cop} to the sum yields
\begin{gather*}
\Delta_4 X^{
\parbox[c][7pt][t]{9pt}{\centering
        \begin{fmffile}{phi4_edge_insertion}
        \begin{fmfgraph}(7,7)
\end{fmfgraph}
        \end{fmffile}}
}_2 = 
\mathbb{I}
\otimes
X^{
\parbox[c][7pt][t]{9pt}{\centering
        \begin{fmffile}{phi4_edge_insertion}
        \begin{fmfgraph}(7,7)
\end{fmfgraph}
        \end{fmffile}}
}_2
+
X^{
\parbox[c][7pt][t]{9pt}{\centering
        \begin{fmffile}{phi4_edge_insertion}
        \begin{fmfgraph}(7,7)
\end{fmfgraph}
        \end{fmffile}}
}_2
\otimes
\mathbb{I} +\\
+
    \frac18    
    \parbox[c][20pt][t]{20pt}{\centering
    \begin{fmffile}{tadpole_fullexample}
    \begin{fmfgraph}(20,20)
\end{fmfgraph}
    \end{fmffile}
    }
\otimes 
\parbox[c][20pt][t]{20pt}{\centering
    \begin{fmffile}{tadpole_fullexample}
    \begin{fmfgraph}(20,20)
\end{fmfgraph}
    \end{fmffile}
    }
+
\frac38
\parbox[c][20pt][t]{20pt}{\centering
        \begin{fmffile}{eye_fullexample}
        \begin{fmfgraph}(20,20)
\end{fmfgraph}
        \end{fmffile}}
\otimes
\parbox[c][20pt][t]{20pt}{\centering
    \begin{fmffile}{tadpole_fullexample}
    \begin{fmfgraph}(20,20)
\end{fmfgraph}
    \end{fmffile}
    }.
\end{gather*}
On the other hand identity \eqref{eqn:sum_different} gives:
\begin{align*}
\Delta_4 X^{
\parbox[c][7pt][t]{9pt}{\centering
        \begin{fmffile}{phi4_edge_insertion}
        \begin{fmfgraph}(7,7)
\end{fmfgraph}
        \end{fmffile}}
}_2 &= 
\sum \limits_{ \substack{ \Gamma \in \mathcal{T}\\ h_1(\Gamma) = 2\\\text{\normalfont{res}}(\Gamma) = 
\parbox[c][7pt][t]{9pt}{\centering
        \begin{fmffile}{phi4_edge_insertion}
        \begin{fmfgraph}(7,7)
\end{fmfgraph}
        \end{fmffile}}
}} \frac{\Delta_D \Gamma}{\left| \text{\normalfont{Aut}}(\Gamma) \right|} = 
\mathbb{I}
\otimes
X^{
\parbox[c][7pt][t]{9pt}{\centering
        \begin{fmffile}{phi4_edge_insertion}
        \begin{fmfgraph}(7,7)
\end{fmfgraph}
        \end{fmffile}}}_2
+
X^{
\parbox[c][7pt][t]{9pt}{\centering
        \begin{fmffile}{phi4_edge_insertion}
        \begin{fmfgraph}(7,7)
\end{fmfgraph}
        \end{fmffile}}}_2
\otimes
\mathbb{I}
+ \\
&+
\sum \limits_{\substack{\gamma = \left(\prod \limits_i \gamma_i\right) \in \mathcal{F}\\ \omega_4(\gamma_i) \leq 0 \\h_1(\gamma) = 1}} 
\sum \limits_{\substack{\widetilde\Gamma \in \mathcal{T}\\h_1(\widetilde\Gamma)=1\\\text{\normalfont{res}}(\widetilde\Gamma) = 
\parbox[c][7pt][t]{9pt}{\centering
        \begin{fmffile}{phi4_edge_insertion}
        \begin{fmfgraph}(7,7)
\end{fmfgraph}
        \end{fmffile}}
 } }
\frac{ \left| \mathcal{I}( \widetilde\Gamma|\gamma ) \right| }{ 
    \left| \text{\normalfont{Aut}}( \gamma ) \right| \left| \text{\normalfont{Aut}}(\widetilde\Gamma) \right|}
    \gamma \otimes \widetilde\Gamma,
\end{align*}
where the trivial summands were already evaluated in 
accordance to the special values given in 
\eqref{eqn:1pisum_special}. 

The only relevant graphs for the sum over the subgraphs $\gamma$ on the 
right hand side are $
\parbox[c][20pt][t]{20pt}{\centering
    \begin{fmffile}{tadpole_fullexample}
    \begin{fmfgraph}(20,20)
\end{fmfgraph}
    \end{fmffile}}$ and 
$\parbox[c][20pt][t]{20pt}{\centering
        \begin{fmffile}{eye_fullexample}
        \begin{fmfgraph}(20,20)
\end{fmfgraph}
        \end{fmffile}}$, because these are 
the only 1PI, $\varphi^4$, one-loop graphs 
with $\omega_4(\gamma)\leq0$. 
For the sum over the cographs $\widetilde\Gamma$ only 
$\parbox[c][20pt][t]{20pt}{
    \begin{fmffile}{tadpole_fullexample}
    \begin{fmfgraph}(20,20)
\end{fmfgraph}
    \end{fmffile}}$ has to be considered, because it is the 
only 1PI, one-loop graph with two external legs.

Calculating the numbers of insertions, 
\begin{align*}
\left| \mathcal{I}\left(
        \parbox[c][20pt][t]{20pt}{
    \begin{fmffile}{tadpole_fullexample}
    \begin{fmfgraph}(20,20)
\end{fmfgraph}
    \end{fmffile}
    }\left| 
 \parbox[c][20pt][t]{20pt}{\centering
\begin{fmffile}{tadpole_fullexample}
    \begin{fmfgraph}(20,20)
\end{fmfgraph}
    \end{fmffile}
    }\right.
\right) \right| &= 2 & 
\left| \mathcal{I}\left(
        \parbox[c][20pt][t]{20pt}{\centering
    \begin{fmffile}{tadpole_fullexample}
    \begin{fmfgraph}(20,20)
\end{fmfgraph}
    \end{fmffile}
    }\left| 
\parbox[c][20pt][t]{20pt}{\centering
        \begin{fmffile}{eye_fullexample}
        \begin{fmfgraph}(20,20)
\end{fmfgraph}
        \end{fmffile}}
\right. \right) \right| &= 24,
\end{align*}
which in this case only depend on the possible external leg permutations and
the orders of the automorphism groups,
\begin{align*}
\left| \text{Aut}\left( \parbox[c][20pt][t]{20pt}{\centering
    \begin{fmffile}{tadpole_fullexample}
    \begin{fmfgraph}(20,20)
\end{fmfgraph}
    \end{fmffile}
    }\right)\right| &= 4 & 
\left| \text{Aut}\left(
\parbox[c][20pt][t]{20pt}{\centering
        \begin{fmffile}{eye_fullexample}
        \begin{fmfgraph}(20,20)
\end{fmfgraph}
        \end{fmffile}}
        \right)\right| &= 16,
\end{align*}
yields the same result as the direct calculation above:
\begin{align*}
&\sum \limits_{\substack{\gamma = \left(\prod \limits_i \gamma_i\right) \in \mathcal{F}\\ \omega_4(\gamma_i) \leq 0 \\h_1(\gamma) = 1}} 
\sum \limits_{\substack{\widetilde\Gamma \in \mathcal{T}\\h_1(\widetilde\Gamma)=1\\\text{\normalfont{res}}(\widetilde\Gamma) = 
\parbox[c][7pt][t]{9pt}{\centering
        \begin{fmffile}{phi4_edge_insertion}
        \begin{fmfgraph}(7,7)
\end{fmfgraph}
        \end{fmffile}}
 } }
\frac{ \left| \mathcal{I}( \widetilde\Gamma|\gamma ) \right| }{ 
    \left| \text{\normalfont{Aut}}( \gamma ) \right| \left| \text{\normalfont{Aut}}(\widetilde\Gamma) \right|}
    \gamma \otimes \widetilde\Gamma = \\ 
&=\frac{2}{4 \cdot 4} ~~
\parbox[c][20pt][t]{20pt}{\centering
    \begin{fmffile}{tadpole_fullexample}
    \begin{fmfgraph}(20,20)
\end{fmfgraph}
    \end{fmffile}
    }
\otimes 
\parbox[c][20pt][t]{20pt}{\centering
    \begin{fmffile}{tadpole_fullexample}
    \begin{fmfgraph}(20,20)
\end{fmfgraph}
    \end{fmffile}
    }
+
\frac{24}{16 \cdot 4} ~~
\parbox[c][20pt][t]{20pt}{\centering
        \begin{fmffile}{eye_fullexample}
        \begin{fmfgraph}(20,20)
\end{fmfgraph}
        \end{fmffile}}
\otimes
\parbox[c][20pt][t]{20pt}{\centering
    \begin{fmffile}{tadpole_fullexample}
    \begin{fmfgraph}(20,20)
\end{fmfgraph}
    \end{fmffile}
    }.
\end{align*}

\section{Diagram generation}
\label{chap:diagram_gen}
\subsection{Overview}
The python program \textbf{feyngen} can
generate $\varphi^k$ for $k \ge 3$, QED and Yang-Mills diagrams ready to be used 
in green's function calculations. One of the main purposes of \textbf{feyngen} is to provide input 
for the coproduct calculation in \textbf{feyncop}.

Developing \textbf{feyngen}, the focus was on the generation 
of Feynman diagrams with comparatively large loop orders. 
For instance, all $130516$ 1PI, $\varphi^4$, $8$-loop diagrams with four external legs can 
be generated, together with their symmetry factor, within eight hours and all $342430$ 1PI, QED, vertex residue type, $6$-loop diagrams (with neglect to Furry's theorem) can be 
generated in three days both on a standard end-user PC. 

Additionally to the computation of non-isomorphic diagrams,
\textbf{feyngen} calculates the symmetry 
factors of the resulting graphs.
Handling of leg-fixed and non-leg-fixed graphs is implemented.
Furthermore, options are available to filter for connected, 1PI, 
vertex-2-connected and tadpole free graphs. 

To achieve the high speed for the computation \textbf{feyngen} 
relies on the established \textbf{nauty} package. 

\subsection{Sketch of the implementation of \textbf{feyngen}}
The graphs in \textbf{feyngen} are represented as 
edge lists or respectively as an ordered sequence 
representing the multiset $\Gamma^{[1]}$ of a graph $\Gamma$ 
as defined in definition \ref{def_feynman_graph}. 
The edges themselves are stored as pairs of vertices.
Vertices are represented by integers $\geq 0$. 
The vertex set is not explicitly stored, but every vertex 
is given implicitly by the edges incident to it. This 
is possible, because vertices with valency $0$ are 
not required. 
Internal and external edges are stored equally. 
External vertices are only characterized by their valency being $1$.

In the first step of the computation of non-isomorphic Feynman 
diagrams, the \textbf{nauty} programs \textbf{geng} and \textbf{multig}, 
whose implementation is discussed in \cite{mckay1998isomorph}, 
are used to compute non-isomorphic multigraphs without 
self-loops. 

\textbf{feyngen} decorates these diagrams with self-loops 
if the generation of tadpole diagrams is requested,
until every vertex has the desired vertex type. 
In the case of theories with $3$-valent vertices, the graph is set up 
appropriately with a certain number of extra legs, for all ``snails'' to 
be generated. 

If the generation of QED or Yang-Mills graphs is required the ways to color and direct the graph's edges as photon,
fermion or ghost edges, to yield a valid QED or Yang-Mills graph according 
to definition \ref{def_feynman_graph} , are computed.

Afterwards, the filters requested by the given options are 
tested on the graph. 
If the graph fulfills all desired properties, the 
algorithm continues. Next, if not otherwise stated by the parameters,
the different topologies resulting from possible 
permutations of the external legs are computed. That means 
the external legs are fixed.

To calculate the order of the automorphism group, the edge and vertex 
colored multigraph must be converted to a vertex colored, 
simple graph. This conversation is achieved by adding 
auxiliary colored vertices, representing the colored or 
multiple edges.

In the last step, using \textbf{nauty}, the order of the automorphism group 
of the graph is computed and it is ensured that by coloring, 
adding self-loops or permuting the external legs no isomorphic graphs 
were generated.

\subsection{Validation}
To check the validity of the diagram generation, the sum of the symmetry factors of a zero dimensional QFT can be used. 

From the generation function of the zero dimensional quantum field theory, for instance 
\begin{align}
    \label{eqn_Z_generating}
    Z_{\varphi^k}( a, \lambda, j ) &:= 
    \int \limits_\mathbb{R} \frac{d \varphi}{\sqrt{2 \pi a}} ~
    e^{ - \frac{\varphi^2}{2 a} + \lambda \frac{\varphi^k}{k!} + j \varphi },
\end{align}
in the case of $\varphi_k$ theories can be used to obtain a formal power series, 
\begin{align}
\label{eqn_Z_sum}
\widetilde{Z}_{\varphi^k}( a, \lambda, j ) &= 
\sum \limits_{l \ge 0} 
\sum \limits_{ \substack{ n,m \ge 0 \\n k + m = 2 l } }
\frac{(2l -1)!!}{n! m! (k!)^n} a^l \lambda^n j^m ,
\end{align}
where the factor in front of $a^L \lambda^N j^M$ gives the sum of the symmetry factors of all $\varphi^k$ graphs with $L$ edges, $N$ vertices and $M$ legs.
In the same way, a formal power series giving the sum of the symmetry factors of classes of QED and Yang-Mills graphs can be obtained.
The numbers for connected diagrams can be obtained by expanding $\log\left( \widetilde{Z} \right) $.
A more detailed account of this procedure is given in \cite[chap. 9]{Itzykson}.

If \textbf{feyngen} calculates all diagrams or all 
connected diagrams of a given class, 
the sum of the symmetry factors is 
compared to an corresponding term in a 
power series, associated to a zero dimensional quantum field theory.

Another validation for the computation of 1PI graphs is possible in conjunction 
with $\textbf{feyncop}$. It is described in section \ref{sec:check_feyncop}.

It should be noted that no explicit tests, which compared the output of other Feynman graph generation tools with the output of \textbf{feyngen}, were performed. The reason for this was the ambiguous choice of a graph labeling.

\section{Coproduct computation}
\label{chap:cop_comp}
\subsection{Overview}
The python program \textbf{feyncop} can 
be used to compute the reduced coproduct $\widetilde{\Delta}_D$ of given 
1PI graphs as defined in \eqref{def:red_cop}. 
The output of \textbf{feyngen} can be piped into 
\textbf{feyncop} to calculate the reduced coproduct 
of all 1PI graphs of a given loop order and residue type.

By default, the subgraphs composed of superficially divergent, 1PI graphs 
of the input graphs are computed and given as output. 
These correspond to the left factor of the 
tensor product originating the coproduct in resemblance to 
definition \ref{def_cop}.
Optionally, the complementary cographs,
giving account to the right factor of the tensor product, can be computed.
Furthermore, there is the option to identify the  
sub- and cographs with unlabeled 
1PI graphs.

Additionally, the input graphs can be filtered for primitive graphs (definition \ref{def:primitive}).

The coproduct calculation does only take the degree of divergence 
obtained by power counting, formulated by the map $\omega_D$ as in 
definition \ref{def_omega_D} into
account. Further information, as gained by Furry's theorem in 
the case of QED, is not used.
\subsection{Sketch of the implementation of \textbf{feyncop}}
First \textbf{feyncop} 
reads the input graphs, piped into the python program, 
in a manner compatible with \textbf{feyngen}. 

Internally, graphs are represented as edge lists as 
in the implementation of \textbf{feyngen}. Optionally, a type can be 
assigned to every edge. By default, a bosonic non-oriented edge 
with weight $2$ is assumed. 

To calculate the coproduct, the subgraphs of the given 
graph are computed. They are tested for one particle irreducibleness 
and superficially divergence of their connected components. 
If required, the edges of the cograph are computed by 
shrinking the internal subgraph edges subsequently as 
described in section \ref{subsec:contractions}.

If the sub- and cographs are to be identified with 
full fledged Feynman graphs, two-valent vertices 
are removed from the cograph and 
external legs are added in accordance to the
vertex types of the graph. 

To identify the subgraphs with unlabeled ones the 
\textbf{nauty} package is used to compute a canonical 
labeling. 
Similar types of resulting tensor products are collected 
and given as output.

\subsection{Validation}
\label{sec:check_feyncop}
The reduced coproduct computation was 
checked for validity by testing the 
results for the coassociativity and the gradedness of the 
coproduct derived in chapter \ref{chap:hopf_algebra}. 

Additionally, the coproduct of a sum of 1PI graphs 
of a given graded class of graphs has been 
checked for accordance with 
corollary \ref{coro:sum_formula}. 
This check also confirms the 
validity of the computation of 1PI diagrams by \textbf{feyngen}.

\section{Manual of \textbf{feyngen}}
\label{sec:manual_feyngen}
\subsection{Overview}
In this section the commands of \textbf{feyngen} together 
with input and output format is described.
Because \textbf{feyngen} is a program optimized to generate 
non-isomorphic high loop Feynman 
diagrams, the only mandatory parameter for \textbf{feyngen} 
is the order in $\hbar$  in 
the perturbation series of the class of diagrams to be generated. 
This corresponds to the number
\begin{align*}
h_1( \gamma ) - f_\gamma + 1,
\end{align*}
where $f_\gamma$ is the number of connected components of the 
graphs $\gamma$ to generate.
For connected graphs, $f_\gamma=1$, this number is equivalent to the loop number $h_1( \gamma )$. 

For example, the call to \textbf{feyngen},
\begin{code}
#\$# ./feyngen 3 4 5
\end{code}
will generate all connected $3$, $4$ and $5$ loop vacuum diagrams and all disconnected diagrams which correspond to these 
loop orders in respect to their order in $\hbar$ in the 
perturbation series of $\varphi^4$ theory.
\subsection{Options and Parameters}
Additionally to the loop number, the generation can be controlled by various options. 
\textbf{feyngen} called with the option \mbox{\textbf{-{}-help}} prints the list of all possible program options together with a short description.
\begin{code}
#\$# ./feyngen --help
\end{code}

For instance, only graphs with certain properties can be generated. 
These restrictions can be set by the following options:
\begin{description}
  \item[-c / -{}-connected]  \hfill \\ Generate only connected graphs.
  \item[-p / -{}-1PI]        \hfill \\ Generate only 1PI graphs.
  \item[-v / -{}-vtx2cntd]   \hfill \\ Generate only 2-vertex connected graphs.
  \item[-t / -{}-notadpoles] \hfill \\ Generate only non-tadpole graphs.
\end{description}
By default, \textbf{feyngen} generates $\varphi^4$ graphs. The four options:
\begin{description}
  \item[-k$\#$ / -{}-valence=$\#$]  \hfill \\ Generate graphs with vertex valence $\#$ for a scalar $\varphi^\#$ theory. 
  \item[-{}-phi34]             \hfill \\ Generate scalar graphs with $3$ or $4$ valent vertices ($ \left[\varphi^3 + \varphi^4 \right]$-theory).
  \item[-{}-qed]             \hfill \\ Generate graphs for QED (with neglect of Furry's theorem).
  \item[-{}-qed\_furry]             \hfill \\ Generate graphs for QED (respecting Furry's theorem).
  \item[-{}-ym]             \hfill \\ Generate graphs for Yang-Mills theory.
\end{description}
can be used to alter the type of graphs generated.
The external leg structure of the graphs is determined by the parameters,
\begin{description}
  \item[-j$\#$ / -{}-ext\_legs=$\#$] \hfill \\ Set the total number of external legs $\#$ of $\varphi^k$ graphs. 
  \item[-b$\#$ / -{}-ext\_boson\_legs=$\#$]  \hfill \\  Set the number of external boson legs $\#$ of QED or Yang-Mills graphs.
  \item[-f$\#$ / -{}-ext\_fermion\_legs=$\#$]  \hfill \\  Set the number of external fermion legs $\#$ of QED or Yang-Mills graphs.
  \item[-g$\#$ / -{}-ext\_ghost\_legs=$\#$]  \hfill \\  Set the number of external ghost legs $\#$ of Yang-Mills graphs.
\end{description}
depending on whether graphs for $\varphi^k$-theory, for QED or for Yang-Mills are generated. 
By default, only graphs without external legs are given as output.

Additionally, the behaviour under graph isomorphisms of the external legs can be controlled:
\begin{description}
  \item[-u / -{}-non\_leg\_fixed]  \hfill \\ Generate non-leg-fixed graphs. External legs of graphs are not considered as fixed if this option is set. This option influences isomorphism testing and symmetry factor calculation of graphs.
\end{description}
If not stated otherwise, leg-fixed diagrams are generated.
\subsection{Output of $\varphi^k$-graphs}
\paragraph{Representation of graphs}
Graphs are represented as edge lists. An 
edge is represented by a pair of vertices 
or a triple of two vertices and a letter, representing the type of 
the edge. The types used by \textbf{feyngen} 
are \texttt{f}, \texttt{A} and \texttt{c} 
representing fermion, gauge boson and ghost edges. 
The pairs or triples are embraced by brackets.
Vertices are labeled by integers. 

So for example 
\begin{align*}
\texttt{[2,3]}& &\text{and}& & \texttt{[6,4,f]}
\end{align*}
represent one edge without specific type between the vertices $2$ and $3$ 
and one fermion edge between the vertices $4$ and $6$.
For $\phi^k$ graphs no edge type will be specified. For 
QED and Yang-Mills graphs fermions and ghosts will be oriented edges and 
gauge bosons will be non oriented edges.
This way, $\texttt{[6,4,f]}$ can be interpreted as an fermion pointing 
from vertex $6$ to vertex $4$.

External edges are not distinguished from internal edges, except 
that they are incident to an external one-valent 
vertex as in definition \ref{def_feynman_graph}.

The edge list is embraced by brackets and prefixed by a \texttt{G} to 
simplify the usage of the output with \textbf{maple}.
This output format of graphs will be referred to as the \texttt{G}-format.
\paragraph{Example}
The graphs
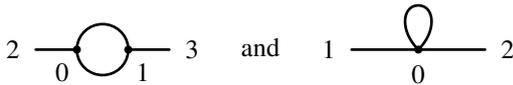
\begin{align*}
   &\parbox[c][25pt][c]{55pt}{\centering
   \begin{fmffile}{example1}
    \begin{fmfgraph*}(50,25)
\end{fmfgraph*}
    \end{fmffile}}& 
    &\text{and}&
   &\parbox[c][25pt][c]{55pt}{\centering
   \begin{fmffile}{example2}
    \begin{fmfgraph*}(50,25)
\end{fmfgraph*}
    \end{fmffile}}
\end{align*}
are represented as
\begin{align*}
  \begin{gathered}
\small \texttt{G[[1,0],[1,0],[2,0],[3,1]]} \\
\text{and} \\
\small \texttt{G[[0,0],[1,0],[2,0]]} ,
  \end{gathered}
\end{align*}
where in the first diagram $0$ and $1$ are the 
internal vertices and $2$ and $3$ are
external vertices. 
In the second diagram $0$ is the only internal vertex with 
$1$ and $2$ depicting external vertices.
The labeling of the vertices is auxiliary and 
is used to give a representative of the isomorphism class of 
the graph. 
The labeling is assigned using \textbf{nauty}, which 
chooses a canonical labeling unique for every isomorphism class 
of graphs. The labeling is chosen, such that the external 
vertices carry the highest labels.

\paragraph{Output of graph sums}
The output format of \textbf{feyngen} is designed to be readable by a \textbf{maple} program. 
Therefore, the graphs generated by \textbf{feyngen} are
written as a sum of graphs with every graph weighted by its 
symmetry factor.

For instance, the sum of all $\varphi^3$ 2-loop graphs without external legs, 
with each graph weighted by its symmetry factor, depicted 
diagrammatically as
\begin{align*}
    &\frac{1}{8}
        \parbox[c][25pt][t]{55pt}{\centering
            \begin{fmffile}{vacuum5_p}
        \begin{fmfgraph*}(50,25)
\end{fmfgraph*}
        \end{fmffile}}
    +
    \frac{1}{12}
        \parbox[c][25pt][t]{45pt}{\centering
            \begin{fmffile}{vacuum4_p}
        \begin{fmfgraph*}(25,25)
\end{fmfgraph*}
        \end{fmffile}}
\end{align*}
can be generated by \textbf{feyngen} as follows: 
\begin{code}
#\$# ./feyngen 2 -k3
phi3_j0_h2 :=
+G[[0,0],[1,0],[1,1]]/8
+G[[1,0],[1,0],[1,0]]/12
;
\end{code}
where the \textbf{2} in the command line stands for diagram
generation of order $\hbar^2$ and \textbf{-k3} for $\varphi^3$-theory, such that 
only graphs with three valent vertices are generated.
\paragraph{Symmetry factors}
As can be seen in the above 
example, the graphs are given weighted by their 
symmetry factor. The format is
\begin{code}
G[...]/Aut,
\end{code}
where \texttt{Aut} is the order of the automorphism group of the graph.

Note that in general, the calculation of the symmetry factor depends
on the \textbf{-u} option if external legs are present, depending 
on whether external legs are fixed or not. 
An explicit example for the behaviour of the \textbf{-u} option is given 
in section \ref{sec:labandunlab}.
\paragraph{Distinguished name for \textbf{maple} usage} 
The sum of graphs is given 
a name, which indicates the loop number(s), the number of 
external edges and the theory type.

That means, the output is always of the 
form:
\begin{code}
phi(k)_j(m)_h(L) := 
+G[...]/Aut1
+G[...]/Aut2
+...
...
;
\end{code}
where  \texttt{(k)} is replaced with the appropriate 
theory degree, \texttt{(m)} is substituted by  
the number of external edges and \texttt{(L)} is 
given by the order in $\hbar$ of the graphs.
\subsection{Output of QED and Yang-Mills graphs}
QED diagrams carry additional information, because they have two different edge types. 
Consequently, an QED edge 
\texttt{[v1,v2,t]} is depicted as a pair of vertices \texttt{v1,v2} together with 
the edge's type \texttt{t}. Possible edge types for 
QED graphs are \texttt{f} for fermions and \texttt{A} for photons. For Yang-Mills 
graphs the possible types are \texttt{f}, \texttt{A} and \texttt{c} for fermions, 
gauge bosons (e.g. gluons) and Faddeev-Popov ghosts respectively.
Fermions and ghosts edges are oriented - gauge boson edges are not.

QED and Yang-Mills Feynman diagrams are represented as edge lists as in the 
last section. Because they are oriented,
fermion and ghost edges are 
depicted as ordered pairs of vertices. 
\paragraph{Example}
For example, the graph
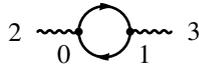
\begin{align*}
   &\parbox[c][25pt][c]{55pt}{\centering
   \begin{fmffile}{example5}
    \begin{fmfgraph*}(50,25)
\end{fmfgraph*}
    \end{fmffile}}
\end{align*}
is represented as
\begin{code}
G[[0,1,f],[1,0,f],[2,0,A],[3,1,A]]
\end{code}
and
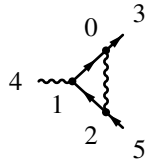
\begin{align*}
   &\parbox[c][35pt][c]{45pt}{\centering
   \begin{fmffile}{example6}
    \begin{fmfgraph*}(35,35)
\end{fmfgraph*}
    \end{fmffile}}
\end{align*}
is depicted as
\begin{code}
G[[1,0,f],[2,1,f],[2,0,A],
    [0,3,f],[5,2,f],[4,1,A]].
\end{code}
The orientation of the fermion lines matches the 
ordering of the vertices in the edges. 
\textbf{feyngen} only generates graphs with valid QED vertex types 
as described in definition \ref{def_feynman_graph}.
\paragraph{Output of graph sums}
The output of graph sums is similar to the output of $\varphi^k$ 
graphs. As for the treatment of $\varphi^k$ graphs, an example 
for the case of QED and Yang-Mills diagram generation is given.
\paragraph{Examples}
Consider the sum of all two loop, photon propagator residue type, 1PI, QED diagrams:
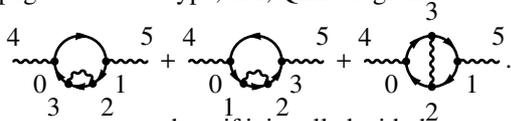
\begin{align*}
   &   \parbox[c][25pt][t]{55pt}{\centering
   \begin{fmffile}{example9}
    \begin{fmfgraph*}(50,25)
\end{fmfgraph*}
    \end{fmffile}} + 
   \parbox[c][25pt][t]{55pt}{\centering
   \begin{fmffile}{example10}
    \begin{fmfgraph*}(50,25)
\end{fmfgraph*}
    \end{fmffile}} +
    \parbox[c][25pt][t]{55pt}{\centering
   \begin{fmffile}{example8}
    \begin{fmfgraph*}(50,25)
\end{fmfgraph*}
    \end{fmffile}}.
\end{align*}
\textbf{feyngen} generates them if it is called with the command line
\begin{code}
#\$# ./feyngen --qed 2 -b2 -p
qed_f0_b2_h2 :=
+G[[0,1,f],[1,2,f],[2,3,f],[3,0,f],
    [3,2,A],[4,0,A],[5,1,A]]/1
+G[[0,1,f],[1,2,f],[2,3,f],[3,0,f],
    [2,1,A],[4,0,A],[5,3,A]]/1
+G[[0,3,f],[1,2,f],[2,0,f],[3,1,f],
    [3,2,A],[4,0,A],[5,1,A]]/1
;
\end{code}
\textbf{-{}-qed} indicates QED graph generation, \textbf{2} 
stands for $2$-loop diagrams ($\hbar^2$), \textbf{-b2} makes \textbf{feyngen} 
generate graphs with 2 photon legs and 
the \textbf{-p} option filters out non 1PI graphs.

For the sum of all one loop, gauge boson propagator residue type, 1PI, Yang-Mills diagrams, 
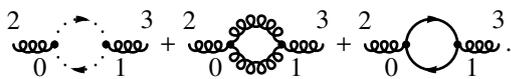
\begin{align*}
   &   \parbox[c][25pt][t]{55pt}{\centering
   \begin{fmffile}{ym_example_1pi_fpg}
    \begin{fmfgraph*}(50,25)
\end{fmfgraph*}
    \end{fmffile}} + 
    \parbox[c][25pt][t]{55pt}{\centering
   \begin{fmffile}{ym_example_1pi_gb}
    \begin{fmfgraph*}(50,25)
\end{fmfgraph*}
    \end{fmffile}} + 
    \parbox[c][25pt][t]{55pt}{\centering
   \begin{fmffile}{ym_example_1pi_f}
    \begin{fmfgraph*}(50,25)
\end{fmfgraph*}
    \end{fmffile}} .
\end{align*}
the call, 
\begin{code}
#\$# ./feyngen --ym 1 -tp -b2 
\end{code}
where the generation of Yang-Mills graphs is triggered with the  \textbf{-{}-ym} option,
gives the desired result:
\begin{code}
ym_f0_g0_b2_h1 :=
+G[[0,1,c],[1,0,c],[2,0,A],[3,1,A]]/1
+G[[0,1,f],[1,0,f],[2,0,A],[3,1,A]]/1
+G[[1,0,A],[1,0,A],[2,0,A],[3,1,A]]/2
;
\end{code}

Another example for Yang-Mills graph generation is the command,
\begin{code}
#\$# ./feyngen --ym 1 -p -f2 -b1
\end{code}
which generates the two one loop, 1PI graphs 
with two fermion and one gluon leg:
\begin{code}
ym_f2_g0_b1_h1 :=
+G[[0,1,f],[2,0,f],[2,1,A],
    [1,3,f],[5,2,f],[4,0,A]]/1
+G[[2,0,f],[1,0,A],[2,1,A],
    [0,3,f],[5,2,f],[4,1,A]]/1
;
\end{code}
Written diagrammatically as, 
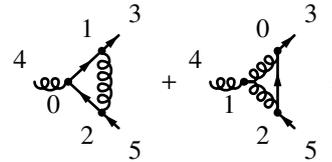
\begin{align*}
   &\parbox[c][45pt][c]{55pt}{\centering
   \begin{fmffile}{ym_example1}
    \begin{fmfgraph*}(35,35)
\end{fmfgraph*}
    \end{fmffile}} + 
   \parbox[c][45pt][c]{55pt}{\centering
   \begin{fmffile}{ym_example2}
    \begin{fmfgraph*}(35,35)
\end{fmfgraph*}
    \end{fmffile}}.
\end{align*}
\paragraph{Distinguished name for \textbf{maple} usage} 
The name of the graph sum for QED diagrams is 
slightly modified:
\begin{code}
qed_f(m1)_b(m2)_h(L) := 
...
;
\end{code}
For the output of graphs respecting Furry's theorem (see \textbf{-{}-qed\_furry} option) the line will read:
\begin{code}
qed_furry_f(m1)_b(m2)_h(L) := 
...
;
\end{code}
Instead of \texttt{j(m)} indicating the total number
of legs as in the case of $\varphi^k$ graphs, 
\texttt{f(m1)} and \texttt{b(m2)} are given 
with \texttt{(m1)} being replaced by 
the number of fermion legs and 
\texttt{(m2)} by the number of photon legs.
For the output of Yang-Mills graphs the pattern is modified in the following way,
\begin{code}
ym_f(m1)_g(m3)_b(m2)_h(L) := 
...
;
\end{code}
with the same numbers and \texttt{(m3)}, the number of ghost legs, to be plugged in.
\subsection{Labeled and unlabeled legs}
\label{sec:labandunlab}
Without the \textbf{-u} option the external 
legs are considered as fixed and leg-fixed diagrams are generated. For this reason
the first and the second graphs in the last QED example are not isomorphic 
and all diagrams carry a symmetry factor of $1$.
The \textbf{-u} option controls this behaviour and influences 
the generation of non-isomorphic graphs and the calculation 
of symmetry factors appropriately.
\paragraph{Example} Consider again the two loop, photon 
propagator, 1PI, QED diagrams, but without fixed 
external legs, 
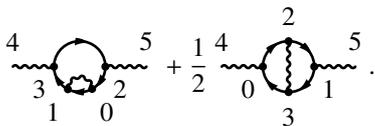
\begin{align*}
   &   \parbox[c][35pt][c]{60pt}{\centering
   \begin{fmffile}{example9_unlab}
    \begin{fmfgraph*}(50,25)
\end{fmfgraph*}
    \end{fmffile}} + 
    \frac12
    \parbox[c][35pt][c]{60pt}{\centering
   \begin{fmffile}{example8_unlab}
    \begin{fmfgraph*}(50,25)
\end{fmfgraph*}
    \end{fmffile}}.
\end{align*}
Because of the additional freedom in permuting the 
edges and vertices, the first two 
diagrams of the last QED example belong to 
the same isomorphism class. 
Furthermore, the second diagram in this example gains an additional 
symmetry generator of index $2$, such that 
the order of the automorphism group increases from $1$ to $2$.
This can be reproduced using \textbf{feyngen} with the \textbf{-u} option:
\begin{code}
#\$# ./feyngen --qed 2 -b2 -p -u
qed_f0_b2_h2_nlf :=
+G[[0,1,f],[1,3,f],[2,0,f],[3,2,f],
    [1,0,A],[4,3,A],[5,2,A]]/1
+G[[0,2,f],[1,3,f],[2,1,f],[3,0,f],
    [3,2,A],[4,0,A],[5,1,A]]/2
;
\end{code}
The text \texttt{nlf}, standing for non-leg-fixed, succeeding
the name of the graph sum acknowledges this behaviour 
for later reproducibility.

\section{Manual of \textbf{feyncop}}
\label{sec:manual_feyncop}
\subsection{Overview}
This section shall introduce the commands and the output 
format of \textbf{feyncop}.

A graph or a sum of graphs, for which the reduced coproduct 
shall be calculated, must 
piped as input into \textbf{feyncop}.
For example,
\begin{code}
#\$# echo "G[[1,0],[1,0],[1,0],[2,0],[3,1]]" 
    | ./feyncop -D4
\end{code}
where the \textbf{echo} command is used to pipe the input into \textbf{feyncop}, 
will give the three proper non empty subgraphs of the diagram 
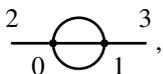
\begin{align*}
    \parbox[c][25pt][t]{55pt}{\centering
    \begin{fmffile}{mirrorsunrise}
    \begin{fmfgraph*}(50,25)
\end{fmfgraph*}
    \end{fmffile}
    },
\end{align*}
which are, in four dimensions, composed of superficially divergent 1PI graphs.
Details to the output format will be given in section \ref{sec:output_sg} and 
the following ones with elaborate examples.

QED graphs are handled the same way, except for 
the weights that must be given with the edges. 
For instance, the subgraphs for the reduced coproduct $\widetilde\Delta_4$ of the graph 
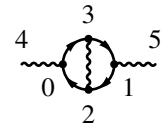
\begin{align*}
    \parbox[c][35pt][c]{55pt}{\centering
   \begin{fmffile}{qed_2l_1pi_prop}
    \begin{fmfgraph*}(50,25)
\end{fmfgraph*}
    \end{fmffile}}
\end{align*}
or given in the \texttt{G}-format as 
\begin{code}
G[[0,3,f],[1,2,f],[2,0,f],[3,1,f],
    [3,2,A],[4,0,A],[5,1,A]]
\end{code}
can be calculated using the command line 
\begin{code}
#\$# echo "G[[0,3,f],[1,2,f],[2,0,f],[3,1,f],
[3,2,A],[4,0,A],[5,1,A]]" | ./feyncop -D4 
\end{code}
.

\subsection{Option and Parameters}
Similar to \textbf{feyngen}, 
\begin{code}
#\$# ./feyncop --help
\end{code}
prints a list of the available options and parameters with a short summary. 

The parameter $D$, altering the dimension parameter of the reduced coproduct $\widetilde \Delta_D$, 
is controlled by the option
\begin{description}
  \item[-D$\#$ / -{}-dimension=$\#$] \hfill \\ Set the dimension for the calculation of the coproduct. 
\end{description}
By default $D=4$ is assumed.

By default, \textbf{feyncop} calculates only the subgraphs, 
which are composed of superficially divergent 1PI graphs, of the input graphs. 
A sum of graphs with a list of the subgraphs, composed of superficially divergent, 1PI graphs is given as output.
This behaviour can be changed by the options:
\begin{description}
  \item[-c / -{}-cographs]              \hfill \\ Calculate and print the cographs additionally to the corresponding subgraphs. 
  Output format: Sum of graphs with a list of the subgraphs, composed of superficially divergent, 1PI graphs, and their cograph.
  \item[-u / -{}-unlabeled]             \hfill \\ Transform the subgraphs and the cographs to unlabeled graphs and identify similar tensor products. Output format: Sum of tensor products.
  \item[-p / -{}-primitives]            \hfill \\ Only filter the input graphs for primitive graphs. Output format: Sum of graphs.
\end{description}
\subsection{Reference to edges}
\textbf{feyncop} represents subgraphs as a 
set of edges of the original graph.
This has the advantage that the location of the 
subgraph in the original graph is implicitly included in the output.
This information is crucial for some applications 
of the coproduct of a Feynman diagram.
The edges are referred to by their 
index in the edge list of the \texttt{G}-format 
starting with $0$.

\paragraph{Example}
The graph, with an auxiliary vertex labeling,
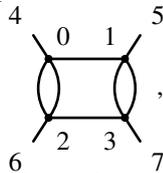
\begin{align*}
    \parbox[c][50pt][c]{50pt}{\centering
    \begin{fmffile}{doubleeye}
    \begin{fmfgraph*}(40,40)
\end{fmfgraph*}
    \end{fmffile}
    },
\end{align*}
represented in the \texttt{G}-format by
\begin{code}
G[[1,0],[2,0],[2,0],[3,1],[3,1],
    [3,2],[4,0],[5,1],[6,2],[7,3]],
\end{code}
will be assigned the following internal edge labels, 
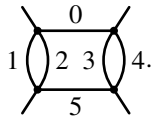
\begin{align*}
    \parbox[c][45pt][c]{50pt}{\centering
    \begin{fmffile}{doubleeye_edgelabeled}
    \begin{fmfgraph*}(40,40)
\end{fmfgraph*}
    \end{fmffile}
    }.
\end{align*}
Here, labels for external legs were omitted for 
simplicity. They are not of interest for 
the description of subgraphs.
\subsection{Output of subgraphs}
\label{sec:output_sg}
By default, only the subgraphs composed of superficially 
divergent 1PI graphs are outputted. 
The output is given as pairs of the original graph 
in the \texttt{G}-format and of the subgraphs, 
each split into its connected components. 
These pairs are preceded by a \texttt{D} to mark a new 
object suitable to be read by \textbf{maple}. 
Therefore, giving a single graph as input, the output will take the form:
\begin{code}
D[ ( original graph ), [ 
{ { 1. subgraph's 1. connected 
    component's edges }, 
{ 1. subgraph's 2. connected 
    component's edges }, 
... 
}, 
{ { 2. subgraph's 1. connected 
    component's edges }, 
    ... }
... 
] ]
\end{code}
\paragraph{Example}
The graph 
    $\parbox[c][20pt][t]{22.5pt}{\centering
    \begin{fmffile}{doubleeye_unlabeled}
    \begin{fmfgraph}(20,20)
\end{fmfgraph}
    \end{fmffile}
    }$,
represented in the \texttt{G}-format as above, 
the output to the call, 
\begin{code}
#\$# echo "G[[1,0],[2,0],[2,0],[3,1],[3,1],[3,2],
[4,0],[5,1],[6,2],[7,3]]" | ./feyncop -D4
\end{code}
will look as follows:
\begin{code}
+ D[G[[1,0],[2,0],[2,0],[3,1],[3,1],
    [3,2],[4,0],[5,1],[6,2],[7,3]],
[{{1,2}}, {{3,4}}, {{1,2},{3,4}}]]
;
\end{code}
The original graph is paired with the 
information about the subgraphs composed of superficially divergent components in 
\texttt{[\{\{1,2\}\}, \{\{3,4\}\}, \{\{1,2\},\{3,4\}\}]}.
The three sets in the list correspond to the three 
subgraphs, indicated by thick lines,
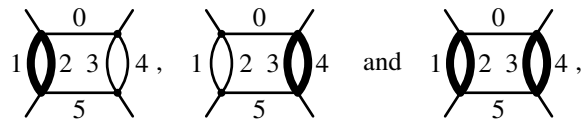
\begin{align*}
&\parbox[c][40pt][t]{55pt}{\centering
    \begin{fmffile}{doubleeye_unlabeled_subgraph1}
    \begin{fmfgraph*}(40,40)
\end{fmfgraph*}
    \end{fmffile}
    },&
&\parbox[c][40pt][t]{55pt}{\centering
    \begin{fmffile}{doubleeye_unlabeled_subgraph2}
    \begin{fmfgraph*}(40,40)
\end{fmfgraph*}
    \end{fmffile}
    }& &\text{and} &
&\parbox[c][40pt][t]{55pt}{\centering
    \begin{fmffile}{doubleeye_unlabeled_subgraph3}
    \begin{fmfgraph*}(40,40)
\end{fmfgraph*}
    \end{fmffile}
    },
\intertext{represented as the sets of sets, }
&\texttt{\{\{1,2\}\}},& &\texttt{\{\{3,4\}\}}& &\text{and}& &\texttt{\{\{1,2\},\{3,4\}\}}.
\end{align*}
The subgraphs are split into their connected components and 
every connected component is given as a set of edge references.

\subsection{Output of cographs}
Called with the \textbf{-c} option, 
\textbf{feyncop} also outputs the 
cographs, obtained by contracting the corresponding subgraphs in
the original graph, for the reduced coproduct. 
With this option the output is a triple 
with the original graph, the subgraphs as described above 
and with the corresponding 
cograph in the \texttt{G}-format. 
The contracted edges of the cographs are not removed from the original edge
list, but replaced by the dummy edge \texttt{[-1,-1]} 
to simplify reference to the edges by index.
The output will be of the form 
\begin{code}
D[ ( original graph ), [ 
[ { { 1. subgraph's 1. connected 
    component's edges }, 
{ 1. subgraph's 2. connected 
    component's edges }, 
... 
}, 
 ( cograph to 1. subgraph ) ], 
[ { { 2. subgraph's 1. connected 
    component's edges }, 
    ... }, 
 ( cograph to 2. subgraph ) ],
... 
] ].
\end{code}

\paragraph{Example}
With $\parbox[c][20pt][t]{22.5pt}{\centering
    \begin{fmffile}{doubleeye_unlabeled}
    \begin{fmfgraph}(20,20)
\end{fmfgraph}
    \end{fmffile}
    }$
as input, 
\begin{code}
#\$# echo "G[[1,0],[2,0],[2,0],[3,1],[3,1],[3,2],
[4,0],[5,1],[6,2],[7,3]]" | ./feyncop -D4 -c
\end{code}
the following outcome 
will be produced:
\begin{code}
+ D[G[[1,0],[2,0],[2,0],[3,1],[3,1],[3,2],
    [4,0],[5,1],[6,2],[7,3]],[
[{{1,2}},
    G[[1,0],[-1,-1],[-1,-1],[3,1],[3,1],
    [3,0],[4,0],[5,1],[6,0],[7,3]]], 
[{{3,4}},
    G[[1,0],[2,0],[2,0],[-1,-1],[-1,-1],
    [1,2],[4,0],[5,1],[6,2],[7,1]]], 
[{{1,2},{3,4}},
    G[[1,0],[-1,-1],[-1,-1],[-1,-1],[-1,-1],
    [1,0],[4,0],[5,1],[6,0],[7,1]]]]
]
;
\end{code}
The cographs, with vertex and edge labeling, corresponding to the subgraphs above
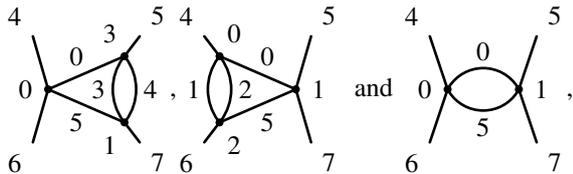
\begin{align*}
&\parbox[c][50pt][c]{60pt}{\centering
    \begin{fmffile}{doubleeye_unlabeled_resgraph1}
    \begin{fmfgraph*}(40,40)
\end{fmfgraph*}
    \end{fmffile}
    },
\parbox[c][50pt][c]{60pt}{\centering
    \begin{fmffile}{doubleeye_unlabeled_resgraph2}
    \begin{fmfgraph*}(40,40)
\end{fmfgraph*}
    \end{fmffile}
    } ~~\text{and}~~
\parbox[c][50pt][c]{60pt}{\centering
    \begin{fmffile}{doubleeye_unlabeled_resgraph3}
    \begin{fmfgraph*}(40,40)
\end{fmfgraph*}
    \end{fmffile}
    },
\end{align*}
are represented in the \texttt{G}-format as
\begin{code}
G[[1,0],[-1,-1],[-1,-1],[3,1],[3,1],[3,0],
    [4,0],[5,1],[6,0],[7,3]]
,
G[[1,0],[2,0],[2,0],[-1,-1],[-1,-1],[1,2],
    [4,0],[5,1],[6,2],[7,1]] 
#\normalfont{and}#
G[[1,0],[-1,-1],[-1,-1],[-1,-1],[-1,-1],[1,0],
    [4,0],[5,1],[6,0],[7,1]].
\end{code}
In the course of the calculation of the cographs, vertices are removed. Therefore, the vertex labeling of the
cographs is not compatible with the one of the original graph. 
Whenever an edge is contracted and 
two vertices are merged, the new vertex will carry the smaller label.

\subsection{Output of tensor products}
Called with the option \textbf{-u}, \textbf{feyncop} 
identifies the subgraphs and cographs 
with unlabeled graphs, groups them in tensor 
products and sums tensor products of similar form. 
Called with this option \textbf{feyncop} 
will handle all graphs as non-leg-fixed graphs. Giving 
leg-fixed graphs as input will result in less terms 
with higher factors than expected.
The tensor products are given as pairs 
of a product of superficially divergent connected components of 
the subgraphs and the cograph. The pairs are preceded 
by an \texttt{T} to indicate the tensor product 
type of output. The output is a sum 
of these tensor products.

The output of a single tensor product will be of the form 
\begin{code}
( factor ) * T[ ( product of subgraphs ), 
    ( cograph ) ].
\end{code}

\paragraph{Example}
Giving $\parbox[c][20pt][t]{22.5pt}{\centering
    \begin{fmffile}{doubleeye_unlabeled}
    \begin{fmfgraph}(20,20)
\end{fmfgraph}
    \end{fmffile}
    }$
 as input,
\begin{code}
#\$# echo "G[[1,0],[2,0],[2,0],[3,1],[3,1],[3,2],
[4,0],[5,1],[6,2],[7,3]]" | ./feyncop -D4 -u
\end{code}
the output will be:
\begin{code}
+ 2/1 * T[ G[[1,0],[1,0],[2,0],
    [3,0],[4,1],[5,1]], 
G[[1,0],[1,0],[2,0],[2,1],
    [3,2],[4,2],[5,0],[6,1]] ]
+ T[ (G[[1,0],[1,0],[2,0],[3,0],
    [4,1],[5,1]])^2,
G[[1,0],[1,0],[2,0],[3,0],[4,1],[5,1]] ]
;
\end{code}
This output corresponds to the 
reduced coproduct calculation, 
\begin{align*}
\widetilde\Delta_4 \left(
    \parbox[c][20pt][t]{22.5pt}{\centering
    \begin{fmffile}{doubleeye_mini_cop1}
    \begin{fmfgraph}(20,20)
\end{fmfgraph}
    \end{fmffile}
    }\right)
   =2
    \parbox[c][20pt][t]{13.5pt}{\centering
    \begin{fmffile}{doubleeye_mini_cop2}
    \begin{fmfgraph}(10,20)
\end{fmfgraph}
    \end{fmffile}
    }
\otimes 
\parbox[c][20pt][t]{22.5pt}{\centering
    \begin{fmffile}{doubleeye_mini_cop3}
    \begin{fmfgraph}(20,20)
\end{fmfgraph}
    \end{fmffile}
    }
 +  \left(
    \parbox[c][20pt][t]{13.5pt}{\centering
    \begin{fmffile}{doubleeye_mini_cop2}
    \begin{fmfgraph}(10,20)
\end{fmfgraph}
    \end{fmffile}
    }
\right)^2
\otimes
    \parbox[c][20pt][t]{13.5pt}{\centering
    \begin{fmffile}{doubleeye_mini_cop2}
    \begin{fmfgraph}(10,20)
\end{fmfgraph}
    \end{fmffile}
    },
\end{align*}
where,  
    $
    2\parbox[c][20pt][t]{13.5pt}{\centering
    \begin{fmffile}{doubleeye_mini_cop2}
    \begin{fmfgraph}(10,20)
\end{fmfgraph}
    \end{fmffile}
    }
\otimes 
\parbox[c][20pt][t]{22.5pt}{\centering
    \begin{fmffile}{doubleeye_mini_cop3}
    \begin{fmfgraph}(20,20)
\end{fmfgraph}
    \end{fmffile}
    },
$
is represented as 
\begin{code}
2/1 * T[ G[[1,0],[1,0],[2,0],
    [3,0],[4,1],[5,1]], 
G[[1,0],[1,0],[2,0],[2,1],
    [3,2],[4,2],[5,0],[6,1]] ]
\end{code}
and
$
\left(
    \parbox[c][20pt][t]{13.5pt}{\centering
    \begin{fmffile}{doubleeye_mini_cop2}
    \begin{fmfgraph}(10,20)
\end{fmfgraph}
    \end{fmffile}
    }
\right)^2
\otimes
    \parbox[c][20pt][t]{13.5pt}{\centering
    \begin{fmffile}{doubleeye_mini_cop2}
    \begin{fmfgraph}(10,20)
\end{fmfgraph}
    \end{fmffile}
    },
$
is denoted as
\begin{code}
T[ (G[[1,0],[1,0],[2,0],[3,0],[4,1],[5,1]])^2, 
    G[[1,0],[1,0],[2,0],[3,0],[4,1],[5,1]] ].
\end{code}
Note that the labelings of the vertices of the subgraphs and cographs 
in this mode of \textbf{feyncop} carry no resemblance to the vertex labeling of the input graph. 
The labelings are chosen using \textbf{nauty}'s
canonical labeling algorithm. Therefore, the graphs are 
representatives of the corresponding graph isomorphism class.
\subsection{Usage in conjunction with \textbf{feyngen}}
\textbf{feyncop} can be used in conjunction with \textbf{feyngen}. 
The coproduct of every input graph is computed and the sum of 
the coproducts in a format depending on the options is given as output. 
To do this the output of \textbf{feyngen} must be piped into \textbf{feyncop} as input.
For instance,
\begin{code}
#\$# ./feyngen 2 -j2 -k4 -pu | ./feyncop -D4
\end{code}
will yield the sum of the reduced coproducts, given in the \texttt{D}-format,
of all non-leg-fixed two loop, $\varphi^4$, 1PI graphs with two external legs 
weighted by the symmetry factor of the original graph,
\begin{code}
+ 1/8 * D[G[[1,0],[1,0],[1,1],[2,0],[3,0]],
[{{2}}, {{0,1}}]]
+ 1/12 * D[G[[1,0],[1,0],[1,0],[2,0],[3,1]],
[{{0,1}}, {{0,2}}, {{1,2}}]]
;
\end{code}
where the relevant graph sum with edge and vertex labels, of which the 
coproduct is calculated, is
\begin{align*}
\frac18
\left(
\parbox[c][50pt][c]{55pt}{\centering
     \begin{fmffile}{double_tadpole_fullexample}
        \begin{fmfgraph*}(40,40)
\end{fmfgraph*}
        \end{fmffile}}
\right)        +
\frac1{12}
\left(
\parbox[c][25pt][t]{70pt}{\centering
     \begin{fmffile}{triple_eye_fullexample}
        \begin{fmfgraph*}(50,25)
\end{fmfgraph*}
        \end{fmffile}}
\right)
\end{align*}
On the other hand, 
\begin{code}
#\$# ./feyngen 2 -j2 -k4 -pu | ./feyncop -D4 -u
\end{code}
will yield the sum of tensor products of unlabeled graphs 
corresponding to the reduced coproduct applied to the sum of 
all two loop, $\varphi^4$, 1PI graphs with two external legs weighted 
by their symmetry factor:
\begin{code} 
phi4_j2_h2_nlf_red_cop_unlab :=
+ 1/8 * T[ G[[0,0],[1,0],[2,0]], 
    G[[0,0],[1,0],[2,0]] ]
+ 3/8 * T[ G[[1,0],[1,0],[2,0],
    [3,0],[4,1],[5,1]], 
G[[0,0],[1,0],[2,0]] ]
;
\end{code}
This corresponds to the calculation 
\begin{align*}
\widetilde\Delta_4 &\left( 
       \frac18
        \parbox[c][20pt][t]{30pt}{\centering
     \begin{fmffile}{double_tadpole2_fullexample}
        \begin{fmfgraph}(20,20)
\end{fmfgraph}
        \end{fmffile}}
+ 
\frac1{12}
\parbox[c][20pt][t]{40pt}{\centering
     \begin{fmffile}{triple_eye2_fullexample}
        \begin{fmfgraph}(30,20)
\end{fmfgraph}
        \end{fmffile}}
        \right) = \\
    &= \frac18 ~~
    \parbox[c][20pt][t]{20pt}{\centering
    \begin{fmffile}{tadpole_fullexample}
    \begin{fmfgraph}(20,20)
\end{fmfgraph}
    \end{fmffile}
    }
\otimes 
\parbox[c][20pt][t]{20pt}{\centering
    \begin{fmffile}{tadpole_fullexample}
    \begin{fmfgraph}(20,20)
\end{fmfgraph}
    \end{fmffile}
    }
+
\frac38 ~~
\parbox[c][20pt][t]{20pt}{\centering
        \begin{fmffile}{eye_fullexample}
        \begin{fmfgraph}(20,20)
\end{fmfgraph}
        \end{fmffile}}
\otimes
\parbox[c][20pt][t]{20pt}{\centering
    \begin{fmffile}{tadpole_fullexample}
    \begin{fmfgraph}(20,20)
\end{fmfgraph}
    \end{fmffile}
    },
\end{align*}
which can be validated using the 
corollary \eqref{coro:sum_formula} as
shown in the example in section \ref{sec:finally_sum_formula}. 
\subsection{Filtering for primitive graphs}
\textbf{feyncop} has the ability 
to filter the input graphs for primitive ones. 
This behaviour is triggered by the \textbf{-p} 
option. 
A convenient usage pattern is to pipe the output of 
\textbf{feyngen}  
into \textbf{feyncop} to obtain a 
set of primitive graphs with the desired properties. 

\paragraph{Example}
If all primitive, QED, vertex diagrams with 
three loops are desired, the call to \textbf{feyngen},
\begin{code}
#\$# ./feyngen 3 --qed -b1 -f2 -p
\end{code}
to generate the $100$ not necessarily primitive QED diagrams 
is needed. 
To filter these diagrams for primitive ones, 
they can be piped into \textbf{feyncop}:
\begin{code}
#\$# ./feyngen 3 --qed -b1 -f2 -p | ./feyncop -p
\end{code}
This will result in the output
\begin{code}
qed_f2_b1_h3_proj_to_prim :=
+ G[[1,4,f],[2,3,f],[3,6,f],[4,5,f],
    [5,0,f],[6,1,f],[4,3,A],[5,2,A],
    [6,0,A],[0,7,f],[8,2,f],[9,1,A]]
+ G[[1,4,f],[2,3,f],[3,6,f],[4,5,f],
    [5,0,f],[6,1,f],[4,2,A],[5,3,A],
    [6,0,A],[0,7,f],[8,2,f],[9,1,A]]
+ G[[1,2,f],[2,6,f],[3,4,f],[4,5,f],
    [5,1,f],[6,0,f],[3,2,A],[4,0,A],
    [6,5,A],[0,7,f],[8,3,f],[9,1,A]]
+ G[[1,6,f],[2,5,f],[3,4,f],[4,1,f],
    [5,0,f],[6,2,f],[3,2,A],[5,4,A],
    [6,0,A],[0,7,f],[8,3,f],[9,1,A]]
+ G[[1,6,f],[2,4,f],[3,2,f],[4,5,f],
    [5,1,f],[6,0,f],[4,0,A],[5,3,A],
    [6,2,A],[0,7,f],[8,3,f],[9,1,A]]
+ G[[1,6,f],[2,5,f],[3,4,f],[4,0,f],
    [5,1,f],[6,3,f],[4,2,A],[5,3,A],
    [6,0,A],[0,7,f],[8,2,f],[9,1,A]]
+ G[[1,4,f],[2,6,f],[3,5,f],[4,0,f],
    [5,1,f],[6,3,f],[4,3,A],[5,2,A],
    [6,0,A],[0,7,f],[8,2,f],[9,1,A]]
;
\end{code}
corresponding to the sum of the seven primitive, three loop, vertex diagrams in QED:

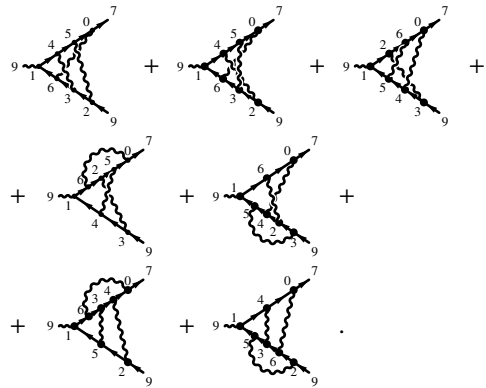
\begin{align*}
   &\tiny \parbox[c][45pt][c]{50pt}{\centering
   \begin{fmffile}{qed_1pi_3l_prim1}
    \begin{fmfgraph*}(35,35)
\end{fmfgraph*}
    \end{fmffile}} +~ 
   \parbox[c][45pt][c]{50pt}{\centering
   \begin{fmffile}{qed_1pi_3l_prim2}
    \begin{fmfgraph*}(35,35)
\end{fmfgraph*}
    \end{fmffile}} +~ 
   \parbox[c][45pt][c]{50pt}{\centering
   \begin{fmffile}{qed_1pi_3l_prim3}
    \begin{fmfgraph*}(35,35)
\end{fmfgraph*}
    \end{fmffile}} +\\
  &+~ \tiny\parbox[c][45pt][c]{50pt}{\centering
   \begin{fmffile}{qed_1pi_3l_prim4}
    \begin{fmfgraph*}(35,35)
\end{fmfgraph*}
   \end{fmffile}} 
 +~\parbox[c][45pt][c]{50pt}{\centering
   \begin{fmffile}{qed_1pi_3l_prim5}
    \begin{fmfgraph*}(35,35)
\end{fmfgraph*}
    \end{fmffile}} +\\
&+~ \tiny\parbox[c][45pt][c]{50pt}{\centering
   \begin{fmffile}{qed_1pi_3l_prim6}
    \begin{fmfgraph*}(35,35)
\end{fmfgraph*}
   \end{fmffile}} 
+~ \parbox[c][45pt][c]{50pt}{\centering
   \begin{fmffile}{qed_1pi_3l_prim7}
    \begin{fmfgraph*}(35,35)
\end{fmfgraph*}
    \end{fmffile}}.
\end{align*}

\section{Conclusion and future prospects}

The program \textbf{feyngen} to generate Feynman graphs of large loop 
orders and the program \textbf{feyncop} to calculate the 
coproduct of given Feynman graphs were presented and validated. 

In a future work, 
\textbf{feyngen} and \textbf{feyncop} could 
be expanded to handle graphs of further types 
of quantum field theories as for instance 
spontaneously 
broken theories. 
The performance of \textbf{feyngen} to generate 
graphs with external legs, QED and Yang-Mills graphs could 
still be increased significantly by taking 
the automorphism groups of the graphs into account 
during the graph generation. 
The implementation relies on Python and it is 
not optimized at all besides the usage of fast 
canonical labeling in the form of the \textbf{nauty} package.

From a larger perspective, the next step
will be to implement a program that evaluates 
the renormalized amplitudes of 1PI graphs from the 
coproduct computed by \textbf{feyncop}
by parametric integration techniques. 
First steps in this endeavour are
pursued in \cite{Panzer:2013cha}.

\section*{Acknowledgments}
I wish to express my grateful thanks to Dirk Kreimer for his great supervision, support and encouragement to work out these results and programs. 
I am also in debt to Oliver Schnetz who provided great advice and steady assistance. He also introduced me to nauty in the first place.

\appendix
\section{Graph theoretic concepts}
\label{sec:graph_basic}
For the graph generation and the calculation of the coproduct, 
the graph theoretical tool \textbf{nauty}, described in \cite{McKay81practicalgraph}, is used. 
Therefore, some definitions of basic graph theoretical terms are required. Here, the focus is on the 
transition from the treatment of graphs in combinatorics and graph theory to a quantum field 
theoretic context. On the graph theoretical side, the definitions are 
based on \cite{bollobas}, whereas the quantum field theoretical view point is 
based on \cite{Manchon} and \cite{KreimerSarsvanSuijlekom}, which can be consulted for a more detailed discussion of 
the following notions. The central objects of 
perturbative QFT are Feynman graphs. These are special cases of multigraphs. 
\subsection{Multigraphs}
\begin{mydef}
A multiset is a generalization of the notion of a set in which elements are allowed to appear more than once.
\end{mydef}
\begin{mydef}
\label{def_graph}
A multigraph $G$ is a pair 
$\left( V, E \right)$, 
where $V$ denotes the vertex set of finite cardinality and $E$ 
is a multiset, the edge multiset, with elements $e \in V \times V$.
\end{mydef}
\paragraph{Adjacency and incidence}
For any edge $e = (v_1,v_2) \in E$ the vertices $v_1$ and $v_2$ are called adjacent to each other and incident to the edge $e$.
\paragraph{Self-loop}
An edge $(v,v)$ incident to only one vertex is called a self-loop.
\paragraph{Valency}
The number of incident edges to a vertex $v \in V$ is called the valency of $v$. 
Self-loops incident to $v$ are counted twice.
\paragraph{Simple graph}
A multigraph $G=(V,E)$ is a simple graph if it has no self-loops and every edge is only present once in $E$.
\paragraph{Directed, undirected and mixed graphs}
Depending on whether the elements of the edge multiset 
$e\in E$ are ordered or unordered pairs of vertices, 
the multigraph is called directed or undirected. 
A multigraph may also contain directed and undirected edges.
\paragraph{Isomorphism}
Two graphs $G=(V,E)$ and $G'=(V',E')$ are considered as isomorphic, $G \simeq G'$, if there is a bijection 
$\phi: G \rightarrow	G'$ 
that maps the vertex set $V$ onto $V'$ and the edge multiset $E$ onto  
$E'$, such that adjacency is respected:
\begin{align*}
\phi(e) &= (\phi(v_1), \phi(v_2)) & & \forall e \in E \text{ with } (v_1,v_2) = e.
\end{align*}
$\phi$ is called an isomorphism between $G$ and $G'$ .
\paragraph{Automorphism group}
The set of all isomorphisms of a graph onto itself $\phi: G \rightarrow G$ is the automorphism group, 
$\text{Aut}(G)$, of $G$.
The inverse of the cardinality of the automorphism group is called the symmetry factor of the graph.

Handling Feynman graphs, a small pathology arises in connection with self-loops. 
For the notions of isomorphism and the automorphism group unoriented self-loops are to be considered as 
edges consisting of two half-edges which can be permuted freely. 
The overall effect of this additional freedom is that every unoriented self-loop
is assigned an additional symmetry generator of the automorphism group of order two.
This generator commutes with all other elements of $\text{Aut}(G)$.
\paragraph{Subgraphs}
A graph $G' = (V',E')$ is a subgraph of $G = (V, E)$ if $E' \subseteq E$ and $V' \subseteq V$.
This relation is denoted as $G' \subseteq G$.
\paragraph{Connectedness}
A multigraph $G$ is disconnected if the vertex set 
and the edge multiset can be split into disjoint sets and multisets $V = V_1 \cup V_2$ and
$E = E_1 \cup E_2$, such that 
$g_1 = \left( V_1, E_1 \right)$ and $g_2 = \left( V_2, E_2 \right)$ are again multigraphs in the sense of Definition \ref{def_graph}. 
Implicitly, this statement is also expressed as $G$ being a disjoint union of the graphs $g_1$ and $g_2$: $G = g_1 \cup g_2$.

Applying this procedure successively, eventually yields 
a set of multigraphs, $\left\{ G_1, \ldots, G_n\right\}$, 
which are connected. These are subgraphs of $G=\bigcup \limits_i G_i$ 
and are called the connected components of $G$.
\paragraph{Loop number or first Betti number}
The loop number or first Betti number $h_1(G)$ is the number of
independent loops or cycles of $G$.

The identity
\begin{align}
\label{eqn:betti}
h_1(G) = |E| - |V| + f_G
\end{align}
is valid for every graph,
where $|E|$ and $|V|$ are the appropriate 
cardinalities of the vertex sets and edge multisets of $G$ 
and $f_G$ denotes the number of connected components of $G$.
\section*{References}

\bibliographystyle{elsarticle-num}
\bibliography{literature}

\end{document}